\documentclass[11pt]{article}

\usepackage[utf8]{inputenc}

\usepackage[a4paper,margin=1in]{geometry}
\usepackage{amsmath,amssymb,amsthm,bm,mathrsfs,mathtools}
\usepackage{booktabs}
\usepackage{graphicx}
\usepackage{subcaption}
\usepackage{float}
\usepackage{hyperref}
\usepackage{setspace}
\usepackage{enumitem}

\hypersetup{colorlinks=true,linkcolor=blue,citecolor=blue,urlcolor=blue}

\newcommand{\safeincludegraphics}[2][]{%
  \IfFileExists{#2}{%
    \includegraphics[#1]{#2}%
  }{%
    \fbox{%
      \begin{minipage}[c][0.28\textheight][c]{0.90\linewidth}
      \centering
      Figure file not found.\\[0.5em]
      \footnotesize The source path is listed in the \LaTeX{} file.
      \end{minipage}%
    }%
  }%
}

\newtheorem{theorem}{Theorem}
\newtheorem{lemma}{Lemma}
\newtheorem{proposition}{Proposition}

\newtheorem{assumption}{Assumption}
\theoremstyle{remark}
\newtheorem{remark}{Remark}

\newcommand{\E}{\mathbb{E}}
\newcommand{\Var}{\mathrm{Var}}
\newcommand{\Cov}{\mathrm{Cov}}
\newcommand{\Normal}{\mathcal{N}}
\newcommand{\toP}{\xrightarrow{p}}
\newcommand{\toD}{\xrightarrow{d}}

\newcommand{\se}{\mathrm{se}}
\newcommand{\Prb}{\mathbb{P}}

\title{When Screening Misleads: A Robust Mendelian Randomization Test for Reliable Causal Discovery}
\author{Litao Jia and Bo Chen$^{*}$ \\
School of Statistics and Data Science, LPMC and KLMDASR, \\
Nankai University, Tianjin, China\\
*Corresponding author: bochen@nankai.edu.cn}
\date{}

\begin{document}
\maketitle

\begin{abstract}
Mendelian Randomization (MR) has been widely used as a standard approach for identifying causal effects in biomedical research. When the true exposure and outcome variables are unknown among many potential candidates, a common but problematic practice is to first screen for associations and then evaluate causal effects only among exposure–outcome pairs that show significant correlations. We demonstrate that the classical MR ratio estimator suffers from severe type I error inflation under this selection procedure. To address this issue, we propose a novel robust MR test that remains valid regardless of the prior association screening step when there is no true causal effect. We show that the proposed test consistently maintains the correct type I error rate, independent of the association test results. Furthermore, our method can incorporate summary statistics from previous association studies to improve the power of causal effect detection. Extensive simulation studies illustrate the advantages of the proposed method compared with the classical MR approach.
\end{abstract}


\section{Introduction}\label{sec:introduction}

Understanding whether an observed association reflects a causal relationship is a central goal in biomedical research. Conventional observational analyses can identify correlations between exposures and disease outcomes, but such associations may arise from confounding, reverse causation, measurement error, or other sources of bias. Mendelian randomization (MR) has become a widely used approach for strengthening causal inference in this setting by using genetic variants as instrumental variables for modifiable exposures \cite{DaveySmithEbrahim2003,DaveySmithHemani2014}. Because germline genetic variants are assigned before disease onset and are often less susceptible to environmental confounding, MR provides a principled strategy for evaluating whether an exposure--outcome association is likely to have a causal interpretation.

A common scientific use of MR is to examine whether a previously reported observational association is causal. For example, MR analyses have been used to evaluate whether circulating C-reactive protein is causally related to coronary heart disease \cite{Wensley2011CRP}, whether high-density lipoprotein cholesterol has a causal protective effect on myocardial infarction \cite{Voight2012HDL}, and whether different blood lipid fractions have causal effects on coronary heart disease \cite{Holmes2015Lipids}. Similar association-to-causality questions have motivated MR studies of body mass index, blood pressure, vitamin D, smoking, alcohol consumption, sleep traits, and many other epidemiological exposures. These examples illustrate a typical workflow: an exposure and an outcome are first found to be associated, and MR is then applied to assess whether the association is consistent with a causal effect.

The availability of genome-wide association study (GWAS) summary statistics has further expanded the scope of MR. Summary-data MR and two-sample MR make it possible to estimate causal effects by combining genetic associations with the exposure and the outcome obtained from independent studies or large consortia \cite{Burgess2013SummarizedData}. Platforms such as MR-Base have enabled systematic causal inference across large collections of phenotypes and outcomes \cite{Hemani2018MRBase}. This development has encouraged increasingly exploratory and high-throughput analyses, including phenome-wide MR studies. For instance, MR-PheWAS designs have been used to scan the potential causal effects of genetically determined vitamin D across hundreds of UK Biobank outcomes \cite{Meng2019VitaminD}. More broadly, large-scale analyses based on GWAS summary data have been used to investigate causal relationships between many risk factors and many common diseases \cite{Zhu2018CommonDiseases}.

These developments make MR a powerful tool for causal discovery, but they also create a statistical challenge that has received comparatively little attention. In exploratory studies, the exposure and outcome of interest may not be specified in advance. Instead, investigators may first screen many candidate exposure--outcome pairs for observational association and then perform MR only for pairs that pass the initial screening step. Such a workflow is scientifically natural: it reflects the common practice of using MR to follow up promising associations. However, it changes the inferential target. The MR test is no longer applied to a randomly chosen exposure--outcome pair; it is applied conditionally on the event that the pair has already shown evidence of association.

This conditioning matters because the observational association statistic and the MR statistic are computed from overlapping variables. Under confounding, the statistic used to screen the exposure--outcome association can be correlated with the subsequent MR estimator. As a result, conditioning on a significant association can distort the null distribution of the MR statistic, even when the marginal MR test is correctly calibrated. The relevant error probability is therefore not the usual unconditional type I error, but the conditional type I error given the screening event. Standard MR procedures, including the Wald ratio and inverse-variance weighted estimators, are not designed to account for this post-screening conditional distribution.

Existing methodological work in MR has addressed several important sources of invalid inference, including weak instruments \cite{BurgessThompson2011WeakInstruments}, invalid instruments and horizontal pleiotropy \cite{Bowden2015MREgger}, and the use of multiple genetic variants with summarized data \cite{Burgess2013SummarizedData}. These developments are essential for robust MR analysis, but they target different problems from the one considered here. Weak-instrument methods address bias and size distortion caused by weak genetic associations with the exposure. Pleiotropy-robust methods address violations of the exclusion restriction. In contrast, the problem studied in this paper arises even when the MR model is correctly specified: the distortion is induced by applying the MR test only after an association-screening step.

In this paper, we study the post-screening behavior of MR tests in a simple structural model that includes genetic instruments, an exposure, an outcome, and an unobserved confounder. We show that the conventional MR statistic can have severe conditional type I error inflation after selecting exposure--outcome pairs based on their observational association. The inflation is not caused by an incorrect marginal null distribution. Instead, it arises because the selected distribution of the MR statistic differs from its unconditional distribution. We derive the joint asymptotic distribution of the association estimator and the MR estimator, characterize the selected null distribution, and identify the covariance term responsible for the distortion.

Motivated by this analysis, we propose a decorrelated MR statistic for valid post-screening causal testing. The key idea is to remove from the MR estimator the component that is first-order correlated with the observational association statistic. Under the Gaussian first-order approximation, this decorrelation makes the corrected statistic asymptotically independent of the association-screening event. We further show how external association information, such as summary statistics from an independent study, can be incorporated to construct a feasible test. The resulting procedure maintains correct type I error conditional on the screening step and can gain power when the external association estimate is sufficiently precise.

The remainder of the paper is organized as follows. Section~\ref{sec:problem} demonstrates the post-screening type I error inflation of the conventional MR statistic and derives the selected null distribution. Section~\ref{sec:joint} establishes the joint asymptotic distribution of the observational association estimator and the MR estimator. Section~\ref{sec:decorrelated} introduces the proposed decorrelated statistics and proves their post-selection validity. Section~\ref{sec:comparison} discusses the relationship between the proposed covariance components and conventional MR variance formulas. Section~\ref{sec:discussion} concludes with implications for reliable causal discovery in screened MR workflows.

\section{The Problem When Screening}\label{sec:problem}

\subsection{Screened MR workflow and conditional inferential target}\label{subsec:screened_workflow}

Let \(\hat\beta_{XY}\) denote the OLS slope from regressing \(Y\) on \(X\) with an intercept. This statistic summarizes the observational exposure--outcome association. In a pre-specified MR analysis, the exposure--outcome pair is fixed before looking at the data, and the causal null hypothesis is tested without conditioning on an earlier association result. In exploratory association-to-causality analyses, however, the second-stage MR analysis is often performed only after the same exposure--outcome pair has passed an observational association screen. The association screen is therefore not an ancillary descriptive step; it determines which causal hypotheses are subsequently tested.

Let
\[
H_0^\gamma:\gamma=\gamma_0
\]
be the causal null hypothesis, and let \(\hat\theta_n\) denote the conventional MR estimator used in the second stage. In this paper \(\hat\theta_n\) is the Wald ratio in the single-SNP case and the IVW estimator in the fixed-\(K\) multi-SNP case. A two-sided association screen can be written as
\begin{equation}
A_n
=
\left\{
\frac{|\hat\beta_{XY}|}{\widehat{\se}(\hat\beta_{XY})}
>
z_{1-\alpha_{\mathrm{scr}}/2}
\right\},
\label{eq:screening_event_tstat}
\end{equation}
where \(\alpha_{\mathrm{scr}}\) is the screening level. Under the first-order approximation used below, this event is equivalently represented as a thresholding event for the standardized association statistic. The causal test is then applied only on the selected sample space \(A_n\).

This changes the relevant error probability. For a two-sided MR test at nominal level \(a\), the usual marginal type I error is
\begin{equation}
\alpha_{\mathrm{unc}}(a)
=
\Prb_{H_0^\gamma}\{\text{MR test rejects }H_0^\gamma\},
\label{eq:unconditional_type1_target}
\end{equation}
whereas the error probability experienced by the screened workflow is
\begin{equation}
\alpha_{\mathrm{sel}}(a)
=
\Prb_{H_0^\gamma}\{\text{MR test rejects }H_0^\gamma\mid A_n\}.
\label{eq:selected_type1_target}
\end{equation}
This is a selective, or post-selection, error probability: the causal null distribution is evaluated after a data-dependent decision has been made to carry out the MR analysis \cite{Freedman1983Screening,FithianSunTaylor2014,TaylorTibshirani2015,LeeSunSunTaylor2016}. The association null and the causal null are different. Even when \(\gamma=\gamma_0=0\), the population observational association \(\beta_{XY}\) can be nonzero because of confounding, so the screened set may contain exposure--outcome pairs that are strongly associated but not causally related.

The distinction between \eqref{eq:unconditional_type1_target} and \eqref{eq:selected_type1_target} matters because the screening statistic and the MR statistic are computed from overlapping variables. The statistic \(\hat\beta_{XY}\) uses \((X,Y)\), and the Wald or IVW estimator also uses the same exposure and outcome through SNP--exposure and SNP--outcome associations. Under confounding, the first-order fluctuations of \(\hat\beta_{XY}\) and \(\hat\theta_n\) are generally correlated. Conditioning on \(A_n\) can therefore change the distribution of \(\hat\theta_n\) among the hypotheses that are actually tested.

\subsection{Bivariate Gaussian geometry of selection}\label{subsec:gaussian_geometry}

The mechanism can be seen most clearly through a two-dimensional Gaussian approximation. The approximation used in this subsection is not an assumption that the structural variables are normally distributed. It is the first-order limiting experiment established in Section~\ref{sec:joint}: under finite fourth moments and relevance conditions,
\begin{equation}
\sqrt n
\begin{pmatrix}
\hat\beta_{XY}-\beta_{XY}\\
\hat\theta_n-\gamma
\end{pmatrix}
\approx
\Normal\left(
\begin{pmatrix}0\\0\end{pmatrix},
\begin{pmatrix}
\omega_{11} & \omega_{12}\\
\omega_{12} & \omega_{22}
\end{pmatrix}
\right).
\label{eq:section2_joint_gaussian}
\end{equation}
The proof of this approximation is based on a multivariate central limit theorem for sample covariance moments and a delta-method expansion of the covariance-ratio estimators; no normality of \(G\), \(U\), \(\varepsilon_1\), or \(\varepsilon_2\) is required.

Under \(H_0^\gamma:\gamma=\gamma_0\), define the standardized association and MR coordinates
\begin{equation}
T_A=\frac{\sqrt n\,\hat\beta_{XY}}{\sqrt{\omega_{11}}}=Z_1+\mu_n,
\qquad
T_M=\frac{\sqrt n(\hat\theta_n-\gamma_0)}{\sqrt{\omega_{22}}}=Z_2,
\label{eq:TA_TM_definition}
\end{equation}
where
\begin{equation}
\mu_n=\frac{\sqrt n\,\beta_{XY}}{\sqrt{\omega_{11}}},
\qquad
\begin{pmatrix}Z_1\\Z_2\end{pmatrix}
\approx
\Normal\left(
\begin{pmatrix}0\\0\end{pmatrix},
\begin{pmatrix}1&\rho\\ \rho&1\end{pmatrix}
\right),
\qquad
\rho=\frac{\omega_{12}}{\sqrt{\omega_{11}\omega_{22}}}.
\label{eq:rho_geometry}
\end{equation}
Thus, under the causal null, the center of the limiting distribution of \((T_A,T_M)\) is \((\mu_n,0)\). The vertical coordinate is centered at zero because the causal null is true; the horizontal coordinate may be shifted away from zero because the observational exposure--outcome association may be generated by confounding.

The bivariate normal density is supported on the whole plane, but its high-density contours are ellipses. In the standardized coordinates above, the parameter \(\rho\) determines the tilt and eccentricity of these contours. When \(\rho=0\), the two coordinates are independent in the limiting Gaussian experiment; in standardized coordinates the contours are circular. In the original unstandardized coordinates, absence of first-order correlation would give axis-aligned ellipses rather than necessarily circles. The important statistical property is independence under the bivariate Gaussian limit, not the visual circular shape itself.

A two-sided association screen and a two-sided causal rejection event are
\begin{equation}
A=\{|T_A|>c\},
\qquad
B=\{|T_M|>q\},
\qquad q=z_{1-a/2}.
\label{eq:A_B_geometry}
\end{equation}
Without screening, the null rejection probability is \(\Prb(B)=a\). After screening, the relevant probability is
\begin{equation}
\Prb(B\mid A)
=
\frac{\Prb(|T_M|>q,\ |T_A|>c)}{\Prb(|T_A|>c)}.
\label{eq:geometry_conditional_probability}
\end{equation}
Geometrically, the vertical lines \(T_A=\pm c\) restrict the second-stage causal test to the two selected tails of the Gaussian cloud, and the horizontal lines \(T_M=\pm q\) mark the MR rejection region. When \(\rho\ne0\), the tilted Gaussian contours make the selected tails contain a different distribution of \(T_M\) than the full cloud. The conditional rejection probability in \eqref{eq:geometry_conditional_probability} can therefore differ from \(a\).

\begin{figure}[htbp]
\centering
\safeincludegraphics[width=0.98\textwidth]{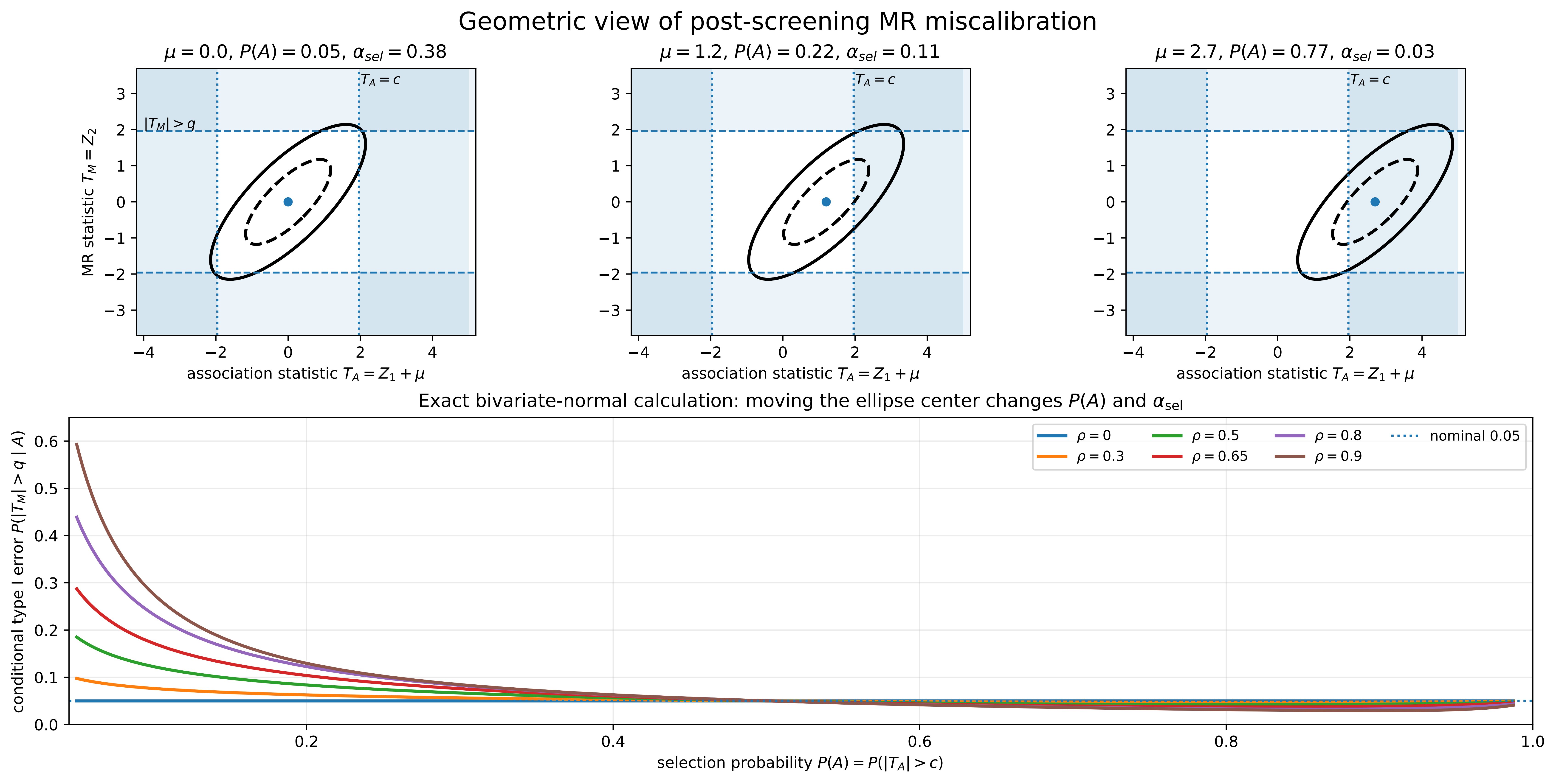}
\caption{Geometric view of post-screening MR miscalibration under the bivariate Gaussian first-order approximation. The top panels show density contours of \((T_A,T_M)\) under \(H_0^\gamma\). The vertical dotted lines represent the association screen \(T_A=\pm c\), and the horizontal dashed lines represent the MR rejection threshold \(T_M=\pm q\). The displayed ellipses are density contours, not boundaries of a uniform distribution. The bottom panel plots the Gaussian probability \(\Prb(|T_M|>q\mid |T_A|>c)\) against the selection probability \(\Prb(|T_A|>c)\) as the horizontal center \(\mu_n\) is varied while \(\rho\) is fixed. The figure is an idealized diagnostic: in the simulations below, varying the confounding strength \(\lambda\) changes the selection rate and can also change the covariance parameters. The purpose of the figure is to show the mechanism by which selective conditioning changes the null distribution, not to reproduce the simulation curves point by point.}
\label{fig:screening_ellipse_visualization}
\end{figure}

The conditional type I error in the Gaussian experiment can be written compactly. For any fixed \(z\), define the selection weight
\begin{equation}
s_{\mu,c,\rho}(z)
=
\Prb(|T_A|>c\mid T_M=z)
=
\Phi\left(\frac{-c-\mu-\rho z}{\sqrt{1-\rho^2}}\right)
+
\Phi\left(\frac{\mu+\rho z-c}{\sqrt{1-\rho^2}}\right),
\label{eq:selection_weight_compact}
\end{equation}
and the marginal selection probability
\begin{equation}
p_{\mu,c}
=
\Prb(|T_A|>c)
=
\Phi(-c-\mu)+\Phi(\mu-c).
\label{eq:selection_probability_compact}
\end{equation}
Then the selected density of \(T_M=Z_2\) is
\begin{equation}
f_{Z_2\mid A}(z)
=
\phi(z)\frac{s_{\mu_n,c,\rho}(z)}{p_{\mu_n,c}},
\label{eq:selected_density_compact}
\end{equation}
and the selected type I error of the uncorrected level-\(a\) MR test is
\begin{equation}
\alpha_{\mathrm{sel}}(a)
=
\int_{|z|>z_{1-a/2}}
\phi(z)\frac{s_{\mu_n,c,\rho}(z)}{p_{\mu_n,c}}\,dz.
\label{eq:selected_alpha_compact}
\end{equation}
Appendix~\ref{app:selected-density} supplements this subsection by deriving \eqref{eq:selected_density_compact} and \eqref{eq:selected_alpha_compact} from the conditional-density identity and the bivariate normal conditioning formula.

Equations~\eqref{eq:selected_density_compact}--\eqref{eq:selected_alpha_compact} also clarify the special case in which screening is harmless. If \(\rho=0\), then \(s_{\mu,c,0}(z)=p_{\mu,c}\) for all \(z\), so \(f_{Z_2\mid A}(z)=\phi(z)\) and \(\alpha_{\mathrm{sel}}(a)=a\). If \(\rho\ne0\), the selection weight varies with \(z\), and the selected density is tilted relative to the marginal standard normal density. The amount of distortion depends on the screening threshold \(c\), the screening strength \(\mu_n\), and the first-order correlation \(\rho\). The ellipse in Figure~\ref{fig:screening_ellipse_visualization} is therefore only a visual representation of Gaussian density contours; the type I error is computed from Gaussian probability mass, not from Euclidean area inside a finite ellipse.

\subsection{Diagnostic simulation}\label{subsec:screening_simulation}

We next demonstrate the same phenomenon numerically under a correctly specified structural MR model,
\begin{equation}
X=\sum_{j=1}^K\alpha_jG_j+U+\varepsilon_1,
\qquad
Y=\gamma X+\lambda U+\varepsilon_2.
\label{eq:model_for_problem}
\end{equation}
Here \(U\) is an unobserved confounder, \(G_1,\ldots,G_K\) are valid genetic instruments, and \(\gamma\) is the causal effect. To isolate type I error, we set \(\gamma=0\). Thus any observational exposure--outcome association is generated by the confounding term \(\lambda U\), not by a causal effect of \(X\) on \(Y\).

The displayed simulations use independent Rademacher instruments \(G_j\in\{-1,1\}\) with \(\Prb(G_j=1)=1/2\), equal first-stage effects \(\alpha_j\equiv1\), and independent centered Gaussian errors. The Gaussian error choice is made for a transparent diagnostic design; it is not required by the asymptotic theory in Section~\ref{sec:joint}. We consider \(K=1\) and \(K=10\), and set
\[
\sigma_U^2=\sigma_{\varepsilon_1}^2=\sigma_{\varepsilon_2}^2
\in\{0.5,1,2\}.
\]
For each parameter setting, the main-sample size is \(n=10{,}000\), the number of Monte Carlo replications is \(R=10{,}000\), the screening test is the two-sided association test for \(\hat\beta_{XY}\) at level \(0.05\), and the second-stage MR test is also carried out at level \(0.05\). The conventional second-stage test uses the package-style Wald statistic for \(K=1\) and the package-style fixed-effect IVW statistic for \(K=10\), without any correction for the screening step. The confounding coefficient \(\lambda\) is varied over a dense grid so that the empirical selection rate ranges from approximately the nominal screening level to almost one.

For each point on the grid, the empirical selection rate is
\[
\widehat p_{\mathrm{sel}}
=
\frac1R\sum_{r=1}^R I(A_n^{(r)}),
\]
the unconditional type I error is
\[
\widehat\alpha_{\mathrm{unc}}
=
\frac1R\sum_{r=1}^R I\{\text{MR rejects in replication }r\},
\]
and the conditional type I error is
\[
\widehat\alpha_{\mathrm{sel}}
=
\frac{\sum_{r=1}^R I\{\text{MR rejects in replication }r\}I(A_n^{(r)})}
{\sum_{r=1}^R I(A_n^{(r)})}.
\]
Plotting the rejection probabilities against \(\widehat p_{\mathrm{sel}}\), rather than against \(\lambda\) itself, puts the results on the scale most relevant for post-screening inference: how informative the selection event is. Unlike the idealized geometry in Figure~\ref{fig:screening_ellipse_visualization}, changing \(\lambda\) changes more than the horizontal location of the Gaussian cloud; it changes the observational association and can also change the covariance entries in the joint approximation. The comparison is therefore qualitative rather than pointwise.

Figure~\ref{fig:conditional_t1e} shows the conditional type I error. The conventional MR statistic can reject far more often than the nominal level after screening. The inflation is most pronounced when the selection event is rare or moderately rare. In this region, the event \(A_n\) is driven by extreme random fluctuations of the observational association statistic; because those fluctuations are correlated with the MR statistic, the selected distribution of the MR statistic is shifted. As \(\lambda\) increases further, selection becomes almost deterministic and conditioning on \(A_n\) carries less information about the random fluctuation of \(\hat\theta_n\); correspondingly, the conditional rejection rate moves back toward the nominal level.

\begin{figure}[htbp]
\centering
\safeincludegraphics[width=0.98\textwidth]{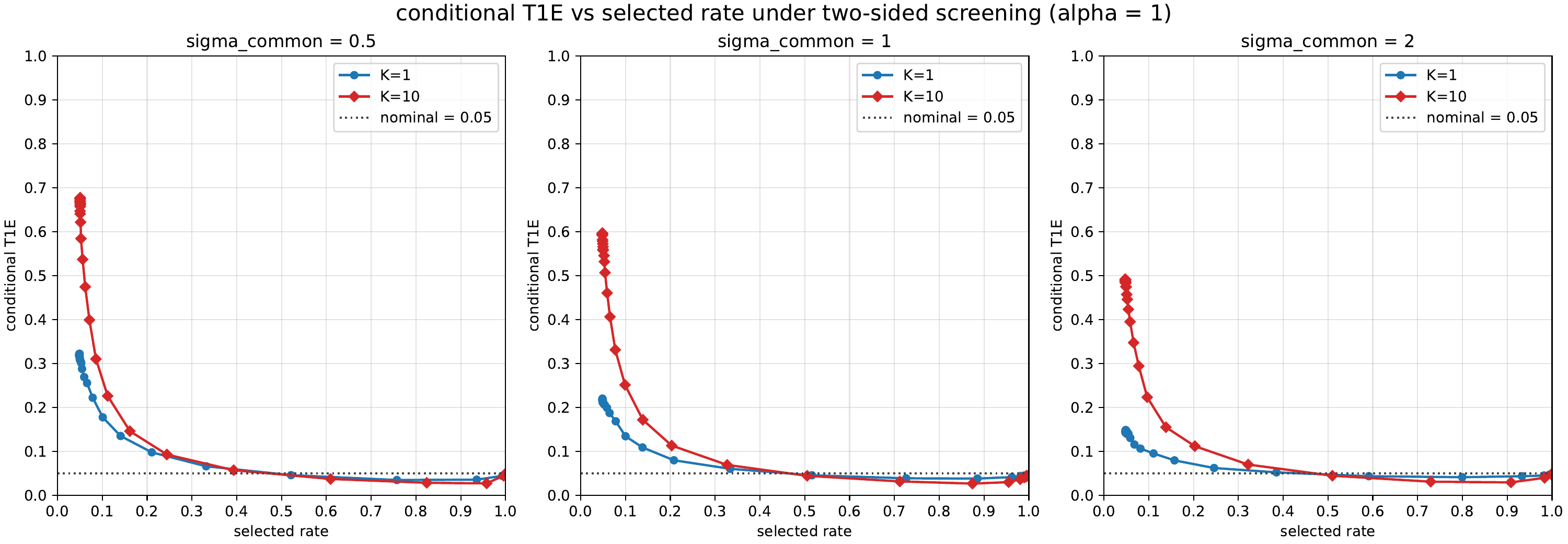}
\caption{Conditional type I error versus empirical selection rate under two-sided observational association screening. The panels correspond to \(\sigma_U^2=\sigma_{\varepsilon_1}^2=\sigma_{\varepsilon_2}^2\in\{0.5,1,2\}\). The nominal level of the second-stage MR test is \(0.05\), shown by the horizontal dotted line. The conventional Wald or IVW statistic can have severe conditional type I error inflation after association screening, even though the instruments are valid and \(\gamma=0\).}
\label{fig:conditional_t1e}
\end{figure}

Figure~\ref{fig:unconditional_t1e} reports the corresponding unconditional type I error under the same simulation settings. The unconditional rejection probability remains close to \(0.05\) across the grid. This contrast is the key diagnostic. The conventional MR statistic is not failing because its marginal null distribution is grossly misspecified in this simulation. It is failing because the statistic is used conditionally on a data-dependent association screen. Therefore a method that only checks marginal calibration cannot diagnose the problem; the relevant target is the conditional distribution induced by \(A_n\).

\begin{figure}[htbp]
\centering
\safeincludegraphics[width=0.98\textwidth]{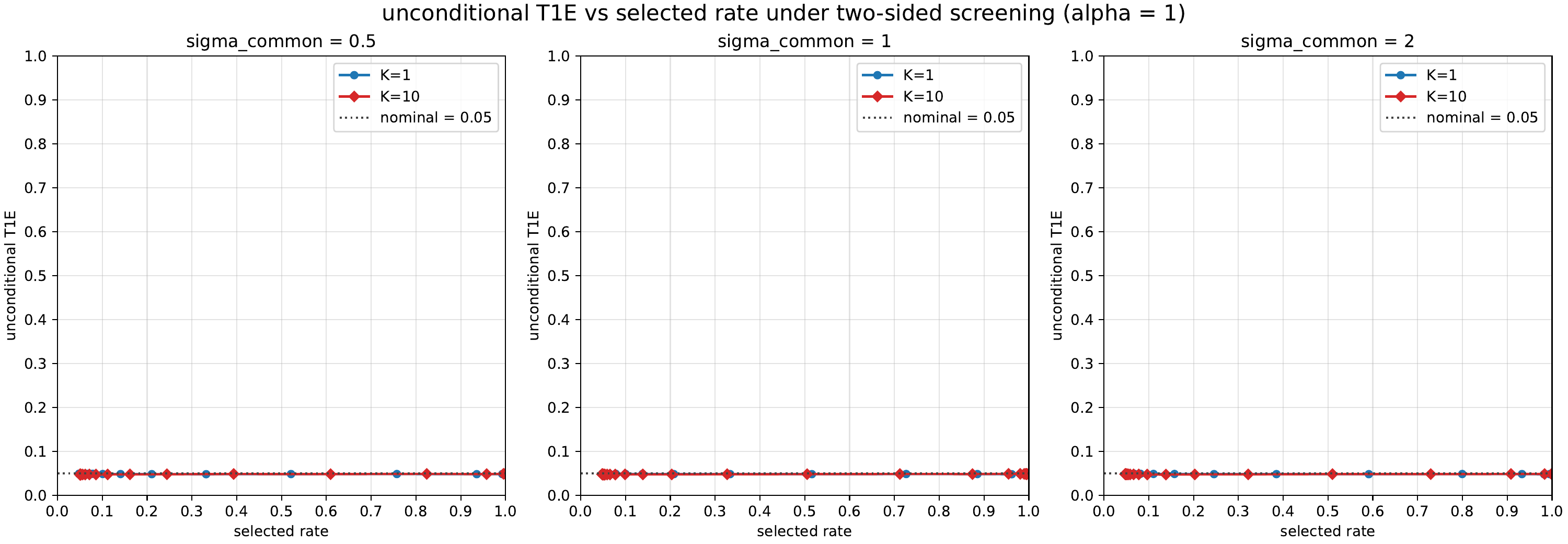}
\caption{Unconditional type I error versus empirical selection rate under the same settings as Figure~\ref{fig:conditional_t1e}. The unconditional rejection probability stays close to the nominal level, showing that the distortion in Figure~\ref{fig:conditional_t1e} is induced by conditioning on the preliminary association screen.}
\label{fig:unconditional_t1e}
\end{figure}

The calculation and simulations suggest the form of a solution. Rather than recalibrating the rejection threshold separately for each possible screening rule, we construct in Section~\ref{sec:decorrelated} a new statistic whose first-order covariance with \(\hat\beta_{XY}\) is zero. Under the Gaussian first-order approximation, this makes the corrected statistic asymptotically independent of the association-screening event, so conditioning on \(A_n\) no longer changes its limiting null distribution.

\section{Joint Asymptotic Foundation}\label{sec:joint}

The selected-distribution argument in Section~\ref{sec:problem} rests on a bivariate Gaussian approximation for the observational association estimator and the MR estimator. This section establishes that approximation and identifies the covariance entries used later for decorrelation. The normality used here is not a distributional assumption on the structural variables. In particular, the genetic variables, the confounder, and the structural errors \(\varepsilon_1\) and \(\varepsilon_2\) need not be normally distributed. The Gaussian law arises as the large-sample limit of estimators that are smooth functions of empirical covariance moments.

The proof strategy is standard. Centered sample covariances are asymptotically equivalent to averages of i.i.d. product variables; the multivariate central limit theorem gives a Gaussian limit for these product averages; and the delta method gives the limit of the covariance-ratio maps defining \(\hat\beta_{XY}\), the Wald ratio, and the IVW estimator \cite{Serfling1980Approximation,VanDerVaart1998,NeweyMcFadden1994,Hampel1974Influence}. Appendix~\ref{app:joint-asymptotics} is the proof appendix for Lemma~\ref{lem:joint-linearization} and Theorems~\ref{thm:single}--\ref{thm:multi}; it gives the sample-moment expansion and the delta-method calculations in detail.

\subsection{Structural model and estimators}\label{subsec:joint_model}

Let \(O_i=(X_i,Y_i,G_{1i},\ldots,G_{Ki})\), \(i=1,\ldots,n\), be i.i.d. observations generated from the fixed-\(K\) structural model
\begin{equation}
X=\sum_{j=1}^K\alpha_jG_j+U+\varepsilon_1,
\qquad
Y=\gamma X+\lambda U+\varepsilon_2.
\label{eq:main_model}
\end{equation}
The single-SNP case is obtained by taking \(K=1\). For readability, the main text uses the zero-mean normalization
\begin{equation}
\E(G_j)=\E(U)=\E(\varepsilon_1)=\E(\varepsilon_2)=0,
\qquad j=1,\ldots,K.
\label{eq:zero_mean_normalization}
\end{equation}
This normalization is not a substantive restriction. In the nonzero-mean case, all formulas are obtained by replacing each variable by its centered version; the sample estimators below are already based on centered sample covariances.

\begin{assumption}[Moment, independence, and relevance conditions]\label{ass:joint}
The following conditions hold throughout Section~\ref{sec:joint}.
\begin{enumerate}[label=(A\arabic*),leftmargin=2.6em]
    \item The observations are i.i.d. copies of \((X,Y,G_1,\ldots,G_K)\), and the number of SNPs \(K\) is fixed as \(n\to\infty\).
    \item The vector \((G_1,\ldots,G_K)\) is independent of \((U,\varepsilon_1,\varepsilon_2)\), and \(U\), \(\varepsilon_1\), and \(\varepsilon_2\) are mutually independent. No normality assumption is imposed. We assume
    \[
    \E\|(G_1,\ldots,G_K)\|^4+
    \E|U|^4+
    \E|\varepsilon_1|^4+
    \E|\varepsilon_2|^4<\infty.
    \]
    \item The SNPs are pairwise uncorrelated, \(\Cov(G_j,G_\ell)=0\) for \(j\ne\ell\), with \(0<\sigma_{G_j}^2=\Var(G_j)<\infty\). This corresponds to the uncorrelated-variant setting of standard fixed-effect IVW estimators; linkage disequilibrium would require replacing the scalar sums below by covariance-matrix expressions.
    \item The exposure variance \(V_X=\Var(X)\) is positive, and the instruments are relevant. In the single-SNP case, \(\Cov(G,X)=\alpha\sigma_G^2\ne0\). In the multi-SNP case, for the limiting IVW weights defined below,
    \[
    \sum_{j=1}^K\pi_j\alpha_j^2>0.
    \]
\end{enumerate}
\end{assumption}

For any two observed variables \(A\) and \(B\), define the centered sample covariance
\begin{equation}
S_{AB}=\frac1n\sum_{i=1}^n(A_i-\bar A)(B_i-\bar B).
\label{eq:sample_covariance}
\end{equation}
The observational association estimator is
\begin{equation}
\hat\beta_{XY}=\frac{S_{XY}}{S_{XX}},
\qquad
\beta_{XY}=\frac{\Cov(X,Y)}{\Var(X)}.
\label{eq:beta_xy_main}
\end{equation}

In the single-SNP case, the causal estimator is the Wald ratio
\begin{equation}
\hat\beta_{MR}=\frac{S_{GY}}{S_{GX}}.
\label{eq:wald_main}
\end{equation}
In the multi-SNP case, define the marginal SNP--exposure and SNP--outcome slopes
\[
\hat\beta_{Xj}=\frac{S_{G_jX}}{S_{G_jG_j}},
\qquad
\hat\beta_{Yj}=\frac{S_{G_jY}}{S_{G_jG_j}},
\qquad j=1,\ldots,K,
\]
and consider the fixed-\(K\) IVW estimator
\begin{equation}
\hat\beta_{IVW}
=
\frac{\sum_{j=1}^K\pi_{j,n}\hat\beta_{Xj}\hat\beta_{Yj}}
{\sum_{j=1}^K\pi_{j,n}\hat\beta_{Xj}^2},
\qquad
\pi_{j,n}\toP\pi_j\in(0,\infty).
\label{eq:ivw_main}
\end{equation}
This is the usual IVW point-estimation form for uncorrelated variants when the weights are chosen as inverse variances of the SNP--outcome associations \cite{Burgess2013SummarizedData}. For the first-order expansion, only the probability limits \(\pi_j\) are needed. Random fluctuations of \(\pi_{j,n}\) do not enter the first-order limit because, under the valid-instrument model, the derivative of the IVW functional with respect to the weights vanishes at the population target \(\gamma\).

\subsection{Asymptotically linear representation}\label{subsec:joint_linearization}

The following lemma records the influence functions needed for the joint limit and for the feasible covariance estimator in Section~\ref{sec:decorrelated}. It avoids a long list of intermediate sample moments in the main text. All variables in the displayed influence functions are population variables under the zero-mean normalization \eqref{eq:zero_mean_normalization}; in the general case they should be replaced by their centered versions.

\begin{lemma}[Asymptotically linear representation]\label{lem:joint-linearization}
Under Assumption~\ref{ass:joint},
\begin{equation}
\sqrt n
\begin{pmatrix}
\hat\beta_{XY}-\beta_{XY}\\
\hat\theta_n-\gamma
\end{pmatrix}
=
\frac1{\sqrt n}\sum_{i=1}^n
\begin{pmatrix}
\phi_{1i}\\
\phi_{2i}
\end{pmatrix}
+o_p(1),
\label{eq:joint_linearization_compact}
\end{equation}
where \(\hat\theta_n=\hat\beta_{MR}\) in the single-SNP case and \(\hat\theta_n=\hat\beta_{IVW}\) in the multi-SNP case. The first component is
\begin{equation}
\phi_1=\frac{X(Y-\beta_{XY}X)}{V_X},
\qquad V_X=\Var(X).
\label{eq:phi1_compact}
\end{equation}
For the single-SNP Wald ratio,
\begin{equation}
\phi_2=\frac{G(Y-\gamma X)}{\Cov(G,X)}
=\frac{G(Y-\gamma X)}{\alpha\sigma_G^2}.
\label{eq:phi2_single_compact}
\end{equation}
For the fixed-\(K\) IVW estimator,
\begin{equation}
\phi_2=
\left(
\sum_{j=1}^K
\frac{\pi_j\alpha_jG_j/\sigma_{G_j}^2}
{\sum_{\ell=1}^K\pi_\ell\alpha_\ell^2}
\right)(Y-\gamma X).
\label{eq:phi2_multi_compact}
\end{equation}
Consequently,
\begin{equation}
\Omega
=
\Var
\begin{pmatrix}
\phi_1\\
\phi_2
\end{pmatrix}
=
\begin{pmatrix}
\omega_{11} & \omega_{12}\\
\omega_{12} & \omega_{22}
\end{pmatrix}
\label{eq:Omega_from_phi_compact}
\end{equation}
is the joint first-order covariance matrix of \((\hat\beta_{XY},\hat\theta_n)\).
\end{lemma}

The proof of Lemma~\ref{lem:joint-linearization} is given in Appendix~\ref{app:joint-asymptotics}. The only probabilistic input is the multivariate central limit theorem for a finite vector of product moments such as \(XY\), \(X^2\), \(G_jY\), and \(G_jX\). The finite-fourth-moment condition in Assumption~\ref{ass:joint} ensures that these product moments have finite variances. Thus the lemma remains valid under non-Gaussian structural errors.

\subsection{Joint limit for the Wald ratio}\label{subsec:single_joint}

\begin{theorem}[Joint limit for the Wald ratio]\label{thm:single}
Assume the single-SNP model
\[
X=\alpha G+U+\varepsilon_1,
\qquad
Y=\gamma X+\lambda U+\varepsilon_2,
\]
and suppose Assumption~\ref{ass:joint} holds with \(K=1\). Then
\begin{equation}
\sqrt n
\begin{pmatrix}
\hat\beta_{XY}-\beta_{XY}\\
\hat\beta_{MR}-\gamma
\end{pmatrix}
\toD
\Normal\left(
\begin{pmatrix}0\\0\end{pmatrix},
\Omega^{S}
\right),
\qquad
\Omega^{S}=
\begin{pmatrix}
\omega_{11} & \omega_{12}\\
\omega_{12} & \omega_{22}^{S}
\end{pmatrix},
\label{eq:single_joint_limit}
\end{equation}
where
\begin{align}
\omega_{11}
&=
\frac{\E\{X^2(Y-\beta_{XY}X)^2\}}{V_X^2},
\label{eq:omega11_single}\\
\omega_{12}
&=
\frac{\lambda^2\sigma_U^2+\sigma_{\varepsilon_2}^2}{V_X}
-
\frac{2\lambda^2\sigma_U^4}{V_X^2},
\label{eq:omega12_single}\\
\omega_{22}^{S}
&=
\frac{\lambda^2\sigma_U^2+\sigma_{\varepsilon_2}^2}{\alpha^2\sigma_G^2}.
\label{eq:omega22_single}
\end{align}
\end{theorem}

The proof of Theorem~\ref{thm:single} is given in Appendix~\ref{app:joint-asymptotics}, Subsection~\ref{app:single-proof}. The theorem should be read with two points in mind. First, the bivariate normal law in \eqref{eq:single_joint_limit} is an estimator limit, not a normal structural model. Second, \(\omega_{11}\) is intentionally left in the robust form \eqref{eq:omega11_single}. Without additional fourth-moment restrictions, \(\E\{X^2(Y-\beta_{XY}X)^2\}\) generally cannot be reduced to a formula involving only second moments. Section~\ref{sec:comparison} and Appendix~\ref{app:heteroskedasticity} discuss why replacing this robust variance target by a classical homoskedastic variance can change the decorrelation coefficient.

\subsection{Joint limit for \texorpdfstring{fixed-\(K\)}{fixed-K} IVW}\label{subsec:multi_joint}

\begin{theorem}[Joint limit for fixed-\(K\) IVW]\label{thm:multi}
Assume model~\eqref{eq:main_model} and Assumption~\ref{ass:joint}. If \(K\) is fixed and \(\pi_{j,n}\toP\pi_j\in(0,\infty)\), then
\begin{equation}
\sqrt n
\begin{pmatrix}
\hat\beta_{XY}-\beta_{XY}\\
\hat\beta_{IVW}-\gamma
\end{pmatrix}
\toD
\Normal\left(
\begin{pmatrix}0\\0\end{pmatrix},
\Omega^{M}
\right),
\qquad
\Omega^{M}=
\begin{pmatrix}
\omega_{11} & \omega_{12}\\
\omega_{12} & \omega_{22}^{M}
\end{pmatrix},
\label{eq:multi_joint_limit}
\end{equation}
where
\begin{align}
\omega_{11}
&=
\frac{\E\{X^2(Y-\beta_{XY}X)^2\}}{V_X^2},
\label{eq:omega11_multi}\\
\omega_{12}
&=
\frac{\lambda^2\sigma_U^2+\sigma_{\varepsilon_2}^2}{V_X}
-
\frac{2\lambda^2\sigma_U^4}{V_X^2},
\label{eq:omega12_multi}\\
\omega_{22}^{M}
&=
(\lambda^2\sigma_U^2+\sigma_{\varepsilon_2}^2)
\frac{\sum_{j=1}^K \pi_j^2\alpha_j^2/\sigma_{G_j}^2}
{\left(\sum_{j=1}^K\pi_j\alpha_j^2\right)^2}.
\label{eq:omega22_multi}
\end{align}
\end{theorem}

The proof of Theorem~\ref{thm:multi} is given in Appendix~\ref{app:joint-asymptotics}, Subsection~\ref{app:multi-proof}. As in the single-SNP case, no normality of \((G_1,\ldots,G_K)\), \(U\), \(\varepsilon_1\), or \(\varepsilon_2\) is required. The closed forms for \(\omega_{12}\) and \(\omega_{22}^{M}\) follow from the linear structural equations, independence, pairwise uncorrelated SNPs, and second moments. The entry \(\omega_{11}\), by contrast, is the robust OLS slope variance and can depend on fourth moments.

\section{Decorrelated Statistics}\label{sec:decorrelated}

Section~\ref{sec:problem} shows that association screening changes the null distribution of a conventional MR statistic whenever the first-order correlation between the screening estimator and the MR estimator is nonzero. Section~\ref{sec:joint} then establishes the joint expansion
\begin{equation}
\sqrt n
\begin{pmatrix}
\hat\beta_{XY}-\beta_{XY}\\
\hat\theta_n-\gamma
\end{pmatrix}
=
\frac1{\sqrt n}\sum_{i=1}^n
\begin{pmatrix}
\phi_{1i}\\
\phi_{2i}
\end{pmatrix}
+o_p(1)
\toD
\Normal\left(
\begin{pmatrix}0\\0\end{pmatrix},
\Omega
\right),
\qquad
\Omega=
\begin{pmatrix}
\omega_{11}&\omega_{12}\\
\omega_{12}&\omega_{22}
\end{pmatrix},
\label{eq:joint_unified_decor}
\end{equation}
where \(\hat\theta_n=\hat\beta_{MR}\) in the single-SNP case and \(\hat\theta_n=\hat\beta_{IVW}\) in the fixed-\(K\) multi-SNP case. The goal of this section is to turn this joint limit into a post-screening valid test of
\[
H_0^\gamma:\gamma=\gamma_0.
\]

\subsection{Population decorrelation and the oracle benchmark}\label{subsec:oracle_benchmark}

The correction is obtained by subtracting from the MR estimator the component that is linearly predictable from the observational association estimator in the joint Gaussian limit. Define
\begin{equation}
    r=\frac{\omega_{12}}{\omega_{11}},
    \qquad
    V_0=\omega_{22}-\frac{\omega_{12}^2}{\omega_{11}}.
\label{eq:r_v0}
\end{equation}
Assuming \(\Omega\) is positive definite, \(\omega_{11}>0\) and \(V_0>0\). The oracle benchmark is
\begin{equation}
U_0
=
\frac{\hat\theta_n-\gamma_0-r(\hat\beta_{XY}-\beta_{XY})}
{\sqrt{V_0/n}}.
\label{eq:U0}
\end{equation}
The numerator in \eqref{eq:U0} is the first-order residual from projecting \(\hat\theta_n\) on \(\hat\beta_{XY}\). Indeed,
\[
\Cov\{\hat\beta_{XY},\hat\theta_n-r\hat\beta_{XY}\}
=
\frac1n(\omega_{12}-r\omega_{11})+o(n^{-1})
=o(n^{-1}).
\]
Under the bivariate Gaussian limit, this zero first-order covariance implies asymptotic independence between the corrected numerator and any screening event determined by \(\hat\beta_{XY}\). This is the key selective-inference idea: validity is assessed after the data-dependent decision to test has been made, rather than under the unconditional sample space \cite{FithianSunTaylor2014,TaylorTibshirani2015,LeeSunSunTaylor2016}.

\begin{theorem}[Oracle decorrelation benchmark]\label{thm:U0}
Suppose the assumptions of Theorem~\ref{thm:single} or Theorem~\ref{thm:multi} hold and \(\Omega\) is positive definite. Under \(H_0^\gamma:\gamma=\gamma_0\),
\[
U_0\toD\Normal(0,1).
\]
Under the local alternative \(\gamma=\gamma_0+\Delta/\sqrt n\),
\[
U_0\toD\Normal\left(\frac{\Delta}{\sqrt{V_0}},1\right).
\]
Moreover, let \(A_n\) be an association-screening event whose limiting form is determined by the first coordinate in \eqref{eq:joint_unified_decor}, with positive limiting probability and boundary probability zero. Then, under the null,
\[
\mathcal L(U_0\mid A_n)\Rightarrow \Normal(0,1).
\]
\end{theorem}

The proof of Theorem~\ref{thm:U0} is given in Appendix~\ref{app:decorrelated}, Subsection~\ref{app:proof-U0}. The theorem is not proposed as a data-analysis procedure because \(U_0\) requires the unknown population association \(\beta_{XY}\) and the unknown covariance entries in \(\Omega\). Its role is to identify the decorrelation direction \(r=\omega_{12}/\omega_{11}\).

\subsection{Using external association information}\label{subsec:external_association}

The oracle statistic subtracts \(r(\hat\beta_{XY}-\beta_{XY})\), but \(\beta_{XY}\) is unknown. Replacing \(\beta_{XY}\) by \(\hat\beta_{XY}\) from the same sample would remove the random quantity that defines the screening event and would not produce a useful causal statistic. The replacement must therefore come from information independent of the main sample used for association screening and MR.

Let \(\hat\beta_{XY}^{(s)}\) be an independent estimator of the same population exposure--outcome association \(\beta_{XY}\). For notational simplicity, suppose it is obtained from an external sample of size \(m\) generated from the same population, so that
\begin{equation}
\sqrt m\{\hat\beta_{XY}^{(s)}-\beta_{XY}\}\toD\Normal(0,\omega_{11}),
\label{eq:external_association_clt}
\end{equation}
independently of the main-sample estimators. In applications, this external input may be available as summary statistics from a previous association study: what is needed is an independent estimate of the observational association between the same exposure and outcome, together with its sampling precision. Standard SNP--exposure and SNP--outcome GWAS summary statistics alone are not sufficient unless they also provide, or can be combined with, such an independent exposure--outcome association estimate.

With this external association estimate, define
\begin{equation}
U_1
=
\frac{\hat\theta_n-\gamma_0-r(\hat\beta_{XY}-\hat\beta_{XY}^{(s)})}
{\sqrt{V_1/n}},
\label{eq:U1}
\end{equation}
where
\begin{equation}
V_1
=
\omega_{22}
+
\left(\frac nm-1\right)\frac{\omega_{12}^2}{\omega_{11}}.
\label{eq:V1}
\end{equation}
The variance in \eqref{eq:V1} is the sum of the oracle residual variance and the additional noise introduced by estimating \(\beta_{XY}\) externally:
\[
\frac{V_0}{n}+r^2\frac{\omega_{11}}{m}
=
\frac1n\left\{\omega_{22}+\left(\frac nm-1\right)\frac{\omega_{12}^2}{\omega_{11}}\right\}.
\]
Thus, if \(m=n\), \(V_1=\omega_{22}\). If \(m>n\), the external association estimate is more precise than the main-sample association estimate, and \(V_1<\omega_{22}\) whenever \(\omega_{12}\ne0\).

\begin{theorem}[External-association benchmark]\label{thm:U1}
Suppose the assumptions of Theorem~\ref{thm:U0} hold. Assume further that \(\hat\beta_{XY}^{(s)}\) satisfies \eqref{eq:external_association_clt}, is independent of the main sample, and \(n/m\to\eta\in(0,\infty)\). Under \(H_0^\gamma:\gamma=\gamma_0\),
\[
U_1\toD\Normal(0,1).
\]
Under the local alternative \(\gamma=\gamma_0+\Delta/\sqrt n\),
\[
U_1\toD\Normal\left(\frac{\Delta}{\sqrt{V_1}},1\right).
\]
Moreover, for association-screening events \(A_n\) satisfying the conditions in Theorem~\ref{thm:U0},
\[
\mathcal L(U_1\mid A_n)\Rightarrow \Normal(0,1)
\]
under the null.
\end{theorem}

The proof of Theorem~\ref{thm:U1} is given in Appendix~\ref{app:decorrelated}, Subsection~\ref{app:proof-U1}. The statistic \(U_1\) is still a benchmark rather than a fully feasible statistic, because \(r\) and \(V_1\) depend on the unknown covariance matrix \(\Omega\).

\begin{remark}[External summary statistics with a reported standard error]
The same construction can be written directly in terms of summary-level external association information. If an independent study reports \(\hat\beta_{XY}^{(s)}\) and a standard error \(\widehat{\se}_s\) for this association estimate, then the external-noise contribution in the denominator is \(r^2\widehat{\se}_s^{\,2}\). Equivalently, one may replace \(V_1/n\) by \(V_0/n+r^2\widehat{\se}_s^{\,2}\), with the corresponding plug-in version used for the feasible statistic below. The main text keeps the equal-population notation in \eqref{eq:external_association_clt}--\eqref{eq:V1} to make the role of the external sample size transparent.
\end{remark}

\subsection{The feasible decorrelated statistic}\label{subsec:feasible_statistic}

The implementable statistic replaces \(\Omega\) by a consistent estimate. Let
\[
\hat\Omega=
\begin{pmatrix}
\hat\omega_{11}&\hat\omega_{12}\\
\hat\omega_{12}&\hat\omega_{22}
\end{pmatrix}
\]
be a consistent estimator of \(\Omega\), and define
\begin{equation}
\hat r=\frac{\hat\omega_{12}}{\hat\omega_{11}},
\qquad
\hat V_1
=
\hat\omega_{22}
+
\left(\frac nm-1\right)\frac{\hat\omega_{12}^2}{\hat\omega_{11}}.
\label{eq:rhat_vhat}
\end{equation}
The feasible decorrelated MR statistic is
\begin{equation}
U_2
=
\frac{\hat\theta_n-\gamma_0-
\hat r(\hat\beta_{XY}-\hat\beta_{XY}^{(s)})}
{\sqrt{\hat V_1/n}}.
\label{eq:U2}
\end{equation}
This is the statistic to be reported in practice. It estimates the same decorrelation direction as \(U_1\), but does so using the empirical influence-function covariance matrix described next.

\subsection{Influence-function covariance estimation}\label{subsec:plugin_estimation}

The estimator of \(\Omega\) should target the joint covariance matrix in \eqref{eq:joint_unified_decor}. We therefore estimate \(\Omega\) from the empirical analogues of the influence functions \(\phi_1\) and \(\phi_2\) in Section~\ref{sec:joint}. This is the usual sandwich-covariance logic for asymptotically linear estimators and general M-estimation \cite{Huber1967,White1980,NeweyMcFadden1994,VanDerVaart1998}.

Let
\[
\tilde X_i=X_i-\bar X,
\qquad
\tilde Y_i=Y_i-\bar Y,
\qquad
\tilde G_i=G_i-\bar G,
\qquad
\tilde G_{ji}=G_{ji}-\bar G_j.
\]
For the observational association estimator \(\hat\beta_{XY}=S_{XY}/S_{XX}\), define
\begin{equation}
\hat\phi_{1i}
=
\frac{\tilde X_i(\tilde Y_i-\hat\beta_{XY}\tilde X_i)}{S_{XX}}.
\label{eq:phi1_hat}
\end{equation}

In the single-SNP case, \(\hat\theta_n=\hat\beta_{MR}=S_{GY}/S_{GX}\), and the sample analogue of the second influence function is
\begin{equation}
\hat\phi_{2i}
=
\frac{\tilde G_i(\tilde Y_i-\hat\beta_{MR}\tilde X_i)}{S_{GX}}.
\label{eq:phi2_single_hat}
\end{equation}

In the fixed-\(K\) multi-SNP case, \(\hat\theta_n=\hat\beta_{IVW}\). Let
\begin{equation}
\hat q_j
=
\frac{\pi_{j,n}\hat\beta_{Xj}/S_{G_jG_j}}
{\sum_{\ell=1}^K\pi_{\ell,n}\hat\beta_{X\ell}^2},
\qquad j=1,\ldots,K,
\label{eq:qhat_multi}
\end{equation}
where \(\hat\beta_{Xj}=S_{G_jX}/S_{G_jG_j}\). Define
\begin{equation}
\hat\phi_{2i}
=
\left(\sum_{j=1}^K\hat q_j\tilde G_{ji}\right)
(\tilde Y_i-\hat\beta_{IVW}\tilde X_i).
\label{eq:phi2_multi_hat}
\end{equation}

Let \(\bar{\hat\phi}_k=n^{-1}\sum_{i=1}^n\hat\phi_{ki}\), \(k=1,2\). The plug-in covariance estimator is
\begin{equation}
\hat\Omega
=
\frac1n\sum_{i=1}^n
\begin{pmatrix}
\hat\phi_{1i}-\bar{\hat\phi}_1\\
\hat\phi_{2i}-\bar{\hat\phi}_2
\end{pmatrix}
\begin{pmatrix}
\hat\phi_{1i}-\bar{\hat\phi}_1 &
\hat\phi_{2i}-\bar{\hat\phi}_2
\end{pmatrix}.
\label{eq:omega_hat_if}
\end{equation}
Equivalently, the three entries \(\hat\omega_{11}\), \(\hat\omega_{12}\), and \(\hat\omega_{22}\) are the empirical variance of \(\hat\phi_{1i}\), the empirical covariance between \(\hat\phi_{1i}\) and \(\hat\phi_{2i}\), and the empirical variance of \(\hat\phi_{2i}\), respectively. Omitting the empirical centering gives the same first-order limit because the population influence functions have mean zero, but the centered version in \eqref{eq:omega_hat_if} is numerically more stable.

\begin{proposition}[Consistency of the influence-function covariance estimator]\label{prop:plugin_consistency}
Suppose the assumptions of Theorem~\ref{thm:single} or Theorem~\ref{thm:multi} hold, the relevant denominators are bounded away from zero with probability tending to one, and \(\pi_{j,n}\toP\pi_j\in(0,\infty)\) in the multi-SNP case. Then
\[
\hat\Omega\toP\Omega.
\]
Consequently,
\[
\hat r\toP r,
\qquad
\hat V_1\toP V_1.
\]
\end{proposition}

The proof of Proposition~\ref{prop:plugin_consistency} is given in Appendix~\ref{app:omega-estimation}. The argument is the standard plug-in consistency proof for an empirical covariance matrix of estimated influence functions: the sample analogues in \eqref{eq:phi1_hat}--\eqref{eq:phi2_multi_hat} converge in mean square to the population influence functions, and the empirical second moments therefore converge to the entries of \(\Omega\).

\begin{center}
\fbox{%
\begin{minipage}{0.92\textwidth}
\textbf{Algorithm 1. Feasible decorrelated MR test after association screening.}

\smallskip
\noindent\textbf{Input:} main-sample observations \((X_i,Y_i,G_i)\) or \((X_i,Y_i,G_{1i},\ldots,G_{Ki})\); an independent external estimate \(\hat\beta_{XY}^{(s)}\) of the observational exposure--outcome association; external sample size \(m\) or an equivalent external standard error; null value \(\gamma_0\).
\begin{enumerate}[leftmargin=2em]
    \item Compute \(\hat\beta_{XY}=S_{XY}/S_{XX}\).
    \item Compute \(\hat\theta_n\): use \(\hat\beta_{MR}=S_{GY}/S_{GX}\) for one SNP and \(\hat\beta_{IVW}\) for multiple SNPs.
    \item Compute \(\hat\phi_{1i}\) from \eqref{eq:phi1_hat} and compute \(\hat\phi_{2i}\) from \eqref{eq:phi2_single_hat} or \eqref{eq:phi2_multi_hat}.
    \item Estimate \(\hat\Omega\) by \eqref{eq:omega_hat_if}; set \(\hat r=\hat\omega_{12}/\hat\omega_{11}\) and \(\hat V_1=\hat\omega_{22}+(n/m-1)\hat\omega_{12}^2/\hat\omega_{11}\).
    \item Report \(U_2\) in \eqref{eq:U2}; reject a two-sided null at level \(a\) if \(|U_2|>z_{1-a/2}\).
\end{enumerate}
\end{minipage}}
\end{center}

\begin{theorem}[Feasible post-screening validity]\label{thm:U2}
Suppose the assumptions of Theorem~\ref{thm:U1} and Proposition~\ref{prop:plugin_consistency} hold. Then
\[
U_2-U_1=o_p(1).
\]
Consequently, under \(H_0^\gamma:\gamma=\gamma_0\),
\[
U_2\toD\Normal(0,1),
\]
and under the local alternative \(\gamma=\gamma_0+\Delta/\sqrt n\),
\[
U_2\toD\Normal\left(\frac{\Delta}{\sqrt{V_1}},1\right).
\]
Moreover, for association-screening events \(A_n\) satisfying the conditions in Theorem~\ref{thm:U0},
\[
\mathcal L(U_2\mid A_n)\Rightarrow \Normal(0,1)
\]
under the null.
\end{theorem}

The proof of Theorem~\ref{thm:U2} is given in Appendix~\ref{app:decorrelated}, Subsection~\ref{app:proof-U2}. This theorem is the formal post-screening result for the implementable statistic: \(U_2\) estimates and removes the first-order component of the MR estimator that is correlated with the observational association statistic, so conditioning on the association screen does not change its limiting null distribution.

\subsection{Power }\label{subsec:decorrelated_power}

The decorrelation is designed for type I error control after screening, but it also changes the local power through the variance factor \(V_1\). For the uncorrected MR statistic
\[
T_R=\frac{\hat\theta_n-\gamma_0}{\sqrt{\omega_{22}/n}},
\]
the local noncentrality parameter under \(\gamma=\gamma_0+\Delta/\sqrt n\) is \(\Delta/\sqrt{\omega_{22}}\). For \(U_2\), the corresponding local noncentrality parameter is \(\Delta/\sqrt{V_1}\). Therefore the asymptotic relative efficiency of the feasible decorrelated test relative to the uncorrected MR test is
\begin{equation}
\mathrm{ARE}_{U_2:R}
=
\frac{\omega_{22}}{V_1}
=
\frac{\omega_{22}}
{\omega_{22}+\left(\frac nm-1\right)\omega_{12}^2/\omega_{11}}.
\label{eq:are}
\end{equation}
Equation~\eqref{eq:are} has a simple interpretation. If \(m>n\), the external association estimate is more precise than the main-sample association estimate, and the corrected statistic can have larger local power than the uncorrected statistic whenever \(\omega_{12}\ne0\). If \(m<n\), the external association estimate is relatively noisy, and the correction may reduce local power. Thus, external association information is not merely a device for feasibility; its precision directly determines the efficiency of the corrected test.

Figure~\ref{fig:conditional_power} illustrates this behavior under alternatives. The conventional package-style statistic may start above the nominal level at \(\gamma=0\), reflecting the conditional type I error inflation described in Section~\ref{sec:problem}. The decorrelated statistics start near the nominal level and then gain power as \(\gamma\) increases. The oracle and external-association benchmarks show the theoretical effect of decorrelation and external-sample precision, while the feasible statistic \(U_2\) shows the performance of the implementable plug-in procedure.

\begin{figure}[htbp]
\centering
\safeincludegraphics[width=0.98\textwidth]{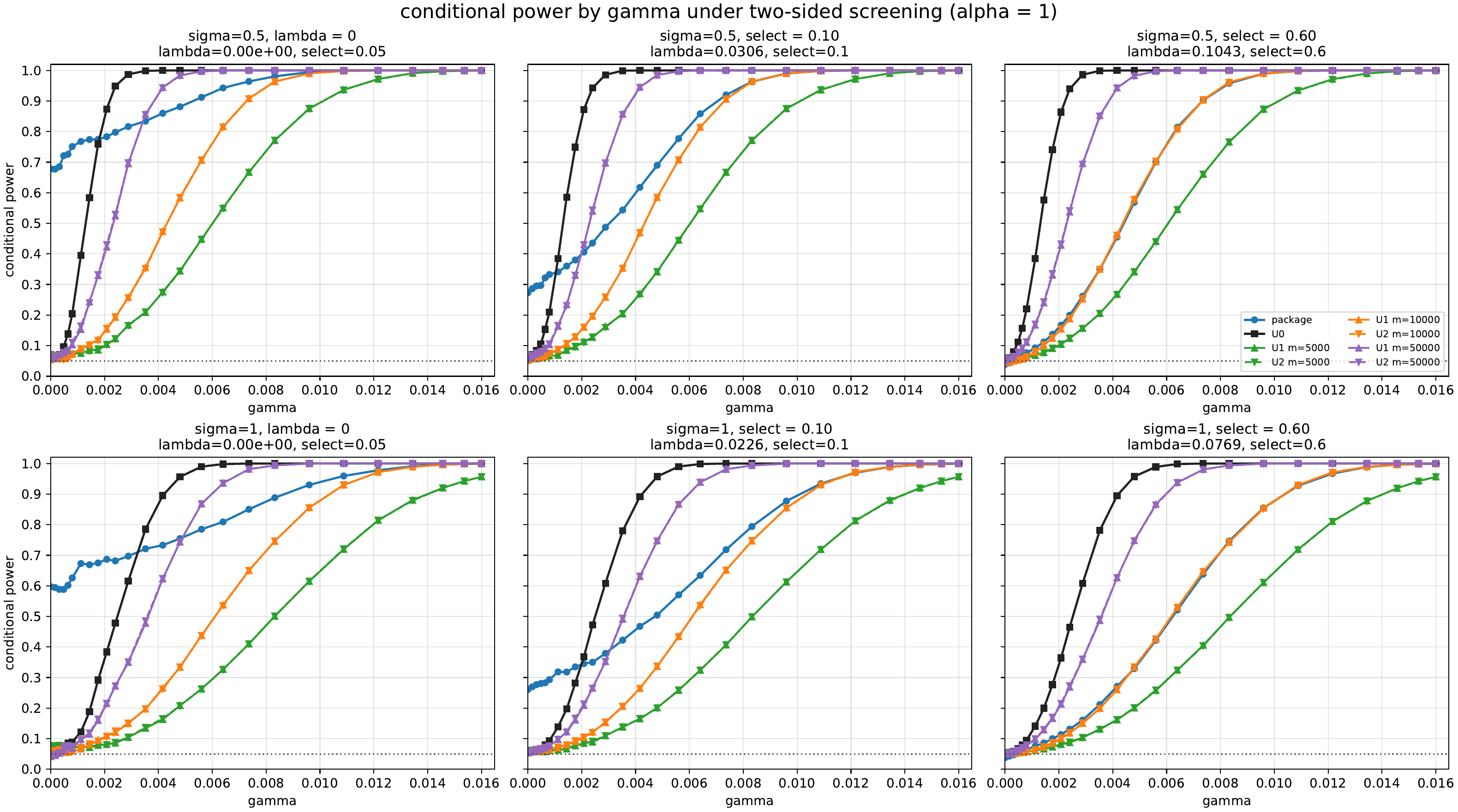}
\caption{Conditional power as a function of \(\gamma\) after two-sided observational association screening. The curves compare the package-style statistic, the oracle benchmark \(U_0\), the external-association benchmark \(U_1\), and the feasible decorrelated statistic \(U_2\) for different external sample sizes. The corrected statistics control the null behavior after screening, and their power improves when the external association estimate is sufficiently precise.}
\label{fig:conditional_power}
\end{figure}

\section{Variance Targets and Implementation Diagnostics}
\label{sec:comparison}

The feasible statistic in Section~\ref{sec:decorrelated} is built from the joint first-order covariance matrix
\[
\Omega=
\begin{pmatrix}
\omega_{11} & \omega_{12}\\
\omega_{12} & \omega_{22}
\end{pmatrix},
\]
where \(\omega_{11}\) is the asymptotic variance factor of the observational association estimator \(\hat\beta_{XY}\), \(\omega_{22}\) is the asymptotic variance factor of the MR estimator \(\hat\theta_n\), and \(\omega_{12}\) is their first-order covariance. These entries have different roles. The decorrelation coefficient is
\[
r=\frac{\omega_{12}}{\omega_{11}},
\]
so \(\omega_{11}\) directly determines whether the corrected numerator is first-order uncorrelated with the screening statistic. By contrast, \(\omega_{22}\) enters only through the variance factor
\[
V_1=\omega_{22}+\left(\frac{n}{m}-1\right)\frac{\omega_{12}^2}{\omega_{11}}.
\]
Thus \(\omega_{22}\) affects standardization, local power, and efficiency, but it does not determine the direction of the decorrelation. This section therefore separates the two issues. Subsection~\ref{subsec:omega11_interpretation} explains why \(\omega_{11}\) must be estimated as a robust random-design slope variance. Subsection~\ref{subsec:omega22_interpretation} explains why different choices of \(\omega_{22}\) can change power in opposite directions for positive and negative causal effects without changing the mechanism that removes post-screening dependence. The algebraic derivations and additional parameter sweeps are deferred to Appendix~\ref{app:heteroskedasticity} for \(\omega_{11}\), Appendix~\ref{app:wald-variance} for the single-SNP ratio variance, and Appendix~\ref{app:ivw-package-variance} for the IVW/package-variance comparison.

\subsection{The robust variance target for \texorpdfstring{\(\omega_{11}\)}{omega11}}
\label{subsec:omega11_interpretation}

Let \(X_c=X-\E(X)\), \(Y_c=Y-\E(Y)\), \(V_X=\Var(X)\), and
\[
\beta_{XY}=\frac{\Cov(X,Y)}{V_X},
\qquad
 e=Y_c-\beta_{XY}X_c.
\]
The observational association estimator \(\hat\beta_{XY}\) is the OLS slope from the random-design linear projection of \(Y\) on \(X\). Its first-order expansion is
\[
\sqrt n(\hat\beta_{XY}-\beta_{XY})
=\frac1{\sqrt n}\sum_{i=1}^n\frac{X_{ci}e_i}{V_X}+o_p(1),
\]
and therefore
\begin{equation}
\omega_{11}=\frac{\E(X_c^2e^2)}{V_X^2}.
\label{eq:omega11_robust_section5}
\end{equation}
This is the heteroskedasticity-robust, or sandwich, variance target for the random-design OLS slope \cite{Huber1967,White1980}. It is not generally equal to the classical homoskedastic regression variance
\begin{equation}
\omega_{11}^{\mathrm{cls}}=\frac{\Var(e)}{V_X}.
\label{eq:omega11_classical_section5}
\end{equation}
The exact difference is
\begin{equation}
\omega_{11}-\omega_{11}^{\mathrm{cls}}
=\frac{\Cov(X_c^2,e^2)}{V_X^2}.
\label{eq:omega11_cov_difference_section5}
\end{equation}
Thus the classical formula is correct only under additional conditions ensuring that the squared projection residual is uncorrelated with the squared exposure. Those conditions are not automatic in the structural model because both \(X\) and \(e\) contain components induced by the instruments and by the confounder.

This distinction matters for the proposed statistic. If the classical target is used in the decorrelation coefficient, the coefficient becomes \(r^{\mathrm{cls}}=\omega_{12}/\omega_{11}^{\mathrm{cls}}\). The first-order covariance left in the corrected numerator is then
\[
\omega_{12}-r^{\mathrm{cls}}\omega_{11}
=
\omega_{12}\left(1-\frac{\omega_{11}}{\omega_{11}^{\mathrm{cls}}}\right),
\]
which is nonzero whenever \(\omega_{11}\ne\omega_{11}^{\mathrm{cls}}\). Consequently, using the classical variance target can fail to remove the first-order dependence on the screening statistic.

The size of the discrepancy can be made explicit in a simple genotype model. In the single-SNP model with \(G\sim\mathrm{Bin}(2,p)\), independent centered normal \(U,\varepsilon_1,\varepsilon_2\), and \(q=p(1-p)\), Appendix~\ref{app:heteroskedasticity} gives
\begin{equation}
\frac{\omega_{11}-\omega_{11}^{\mathrm{cls}}}{\omega_{11}^{\mathrm{cls}}}
=
\frac{2\alpha^4\lambda^2\sigma_U^4q(1-6q)}
{V_X^2\{(\lambda^2\sigma_U^2+\sigma_{\varepsilon_2}^2)V_X-
\lambda^2\sigma_U^4\}}.
\label{eq:omega11_relative_difference_section5}
\end{equation}
The denominator in \eqref{eq:omega11_relative_difference_section5} is positive, so the sign is governed by \(q(1-6q)\). On the minor-allele-frequency scale \(p\in(0,1/2]\), the only nondegenerate zero is
\[
p_0=\frac{3-\sqrt3}{6}\approx0.2113.
\]
Therefore the robust variance is larger than the classical variance for \(0<p<p_0\), smaller for \(p_0<p\le1/2\), and equal at \(p=p_0\). The multi-SNP analogue replaces the single-SNP fourth-power term \(\alpha^4\) by a concentration measure such as \(\sum_j\alpha_j^4\); hence the discrepancy is not restricted to a single-SNP design, although it is diluted when the first-stage signal is spread evenly over many variants.

The parameter regimes that produce a large difference are not purely artificial. Low-frequency and rare variants fall in the region \(p<p_0\), and variants with minor allele frequency thresholds such as 1\% or 5\% are commonly retained or separately analyzed in modern GWAS pipelines \cite{Marees2018GWAS,Uffelmann2021GWAS}. The discrepancy is also amplified when the genetic contribution to \(X\) is strong relative to residual exposure noise and when confounding is substantial. These are precisely the settings in which MR investigators often prefer strong instruments and in which an association-screened workflow is most relevant. A single extremely strong rare variant is not representative of every MR study, but strong cis instruments, biomarker-associated variants, metabolite QTLs, pharmacogenetic variants, and concentrated genetic scores can make such parameter regimes plausible. Since the robust estimator in \eqref{eq:omega11_robust_section5} remains valid when the discrepancy is small, while the classical estimator is valid only under additional restrictions, the robust target is the appropriate default.

Figure~\ref{fig:omega11_diagnostics} summarizes the diagnostic simulations. The grid uses \(G\sim\mathrm{Bin}(2,p)\), \(R=10{,}000\) Monte Carlo replications, \(n=10{,}000\), \(m=5{,}000\), \(\lambda=2\), \(\alpha\in\{2,3\}\), and a dense grid of allele frequencies. The left panel shows that the relative discrepancy can be large at low allele frequencies and changes sign near \(p_0\). The right panel shows the practical consequence: the robust U0 and U2 implementations remain close to the nominal level, whereas replacing \(\omega_{11}\) by \(\omega_{11}^{\mathrm{cls}}\) can lead to severe rejection-rate inflation when the discrepancy is large. For example, in the most extreme reported grid point, the robust U2 rejection rate is about 0.054, whereas the classical U2 rejection rate is about 0.526. Even at \(p=0.05\) in a stress-test setting, the classical U2 rejection rate is about 0.141 while the robust U2 rejection rate remains close to 0.05.

\begin{figure}[htbp]
\centering
\begin{subfigure}{0.88\textwidth}
\centering
\safeincludegraphics[width=\textwidth]{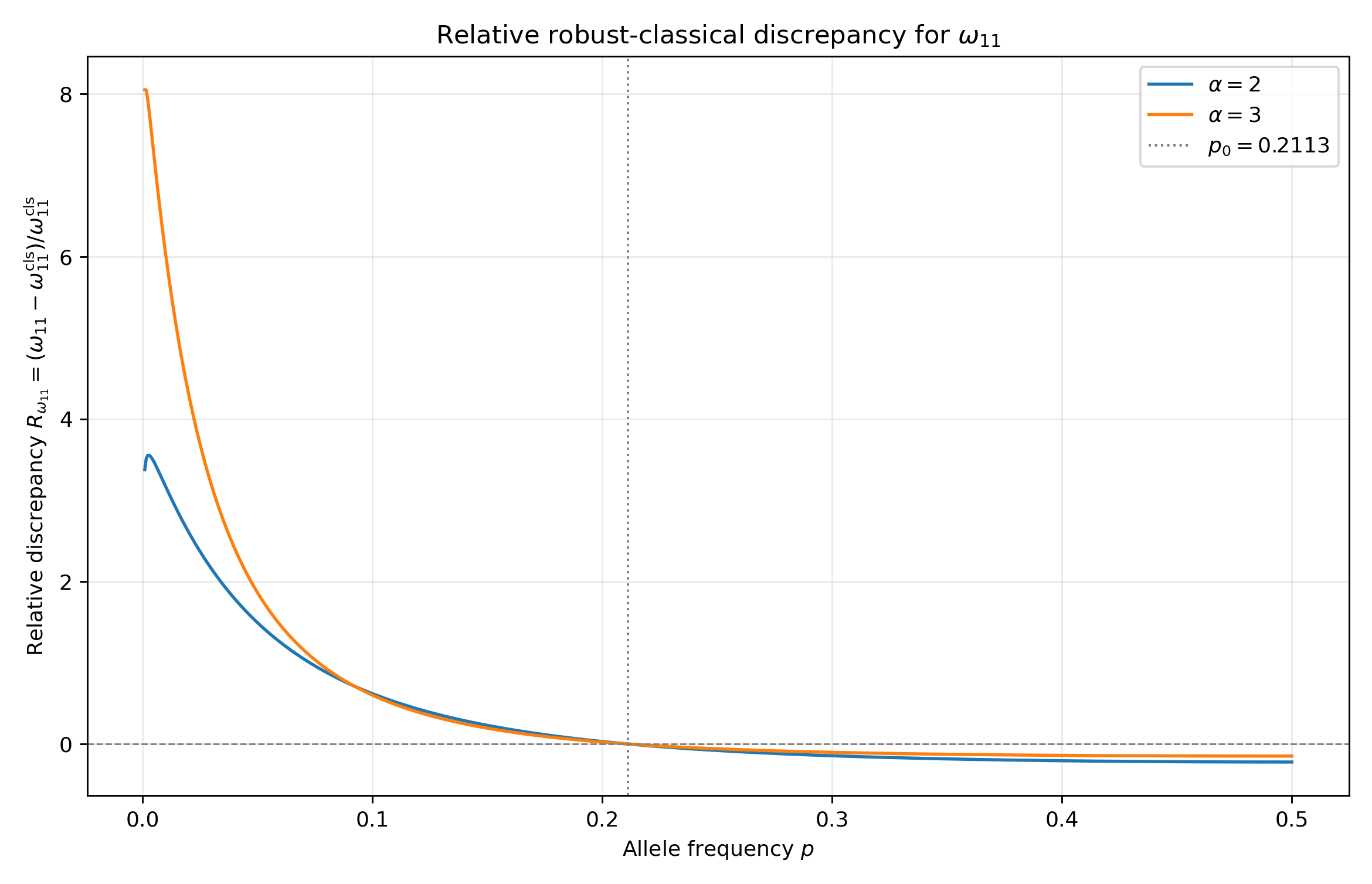}
\caption{Relative discrepancy \(R_{\omega_{11}}=(\omega_{11}-\omega_{11}^{\mathrm{cls}})/\omega_{11}^{\mathrm{cls}}\) as a function of allele frequency.}
\end{subfigure}
\vspace{0.8em}
\begin{subfigure}{0.88\textwidth}
\centering
\safeincludegraphics[width=\textwidth]{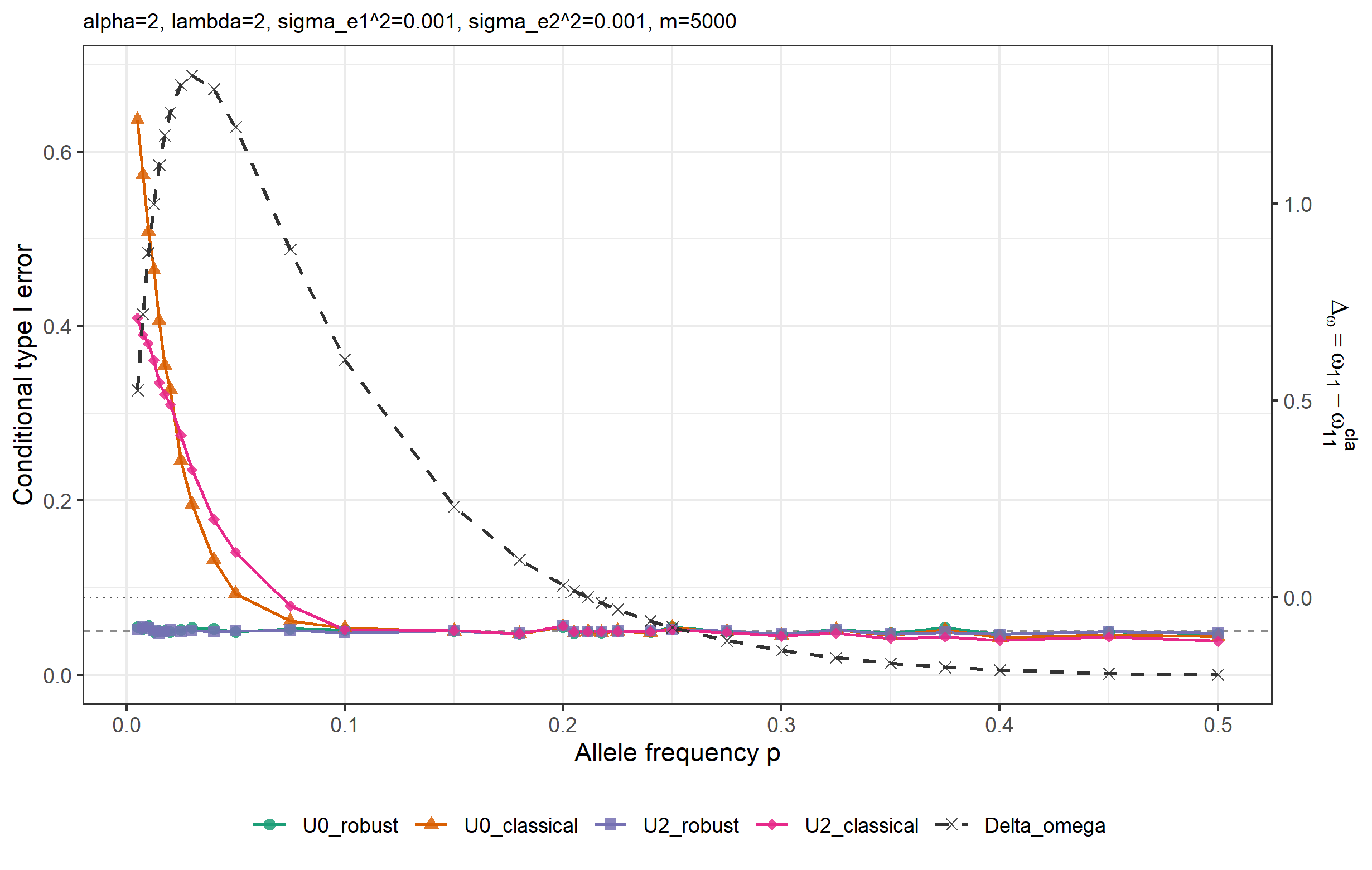}
\caption{Rejection probability of robust and classical implementations as the discrepancy changes.}
\end{subfigure}
\caption{Diagnostics for the \(\omega_{11}\) variance target. The robust target is the influence-function variance in \eqref{eq:omega11_robust_section5}; the classical target is \eqref{eq:omega11_classical_section5}. Large robust--classical discrepancies occur in parameter regimes that include low-frequency instruments and concentrated first-stage signal. In those regimes, using the classical target can make the proposed decorrelation ineffective, whereas the robust implementation remains calibrated.}
\label{fig:omega11_diagnostics}
\end{figure}

\subsection{The role of \texorpdfstring{\(\omega_{22}\)}{omega22}: standardization rather than decorrelation}
\label{subsec:omega22_interpretation}

The interpretation of \(\omega_{22}\) is different. It does not enter the decorrelation coefficient \(r=\omega_{12}/\omega_{11}\). Therefore an alternative estimate of \(\omega_{22}\) can change the scale and power of the statistic, but it does not determine whether the numerator has been made first-order uncorrelated with \(\hat\beta_{XY}\). This is why the \(\omega_{22}\) simulations should be read as variance and power diagnostics rather than as evidence about the validity of the decorrelation step.

For the single-SNP Wald ratio, the full delta-method variance includes the three usual ratio terms,
\begin{equation}
\Var(\hat\beta_{MR})
\approx
\frac{\Var(\hat\beta_{GY})}{\beta_{GX}^2}
+
\frac{\beta_{GY}^2}{\beta_{GX}^4}\Var(\hat\beta_{GX})
-
2\frac{\beta_{GY}}{\beta_{GX}^3}\Cov(\hat\beta_{GY},\hat\beta_{GX}),
\label{eq:wald_full_delta_section5}
\end{equation}
which is the standard delta-method expansion for a ratio estimator \cite{Fieller1954,Cochran1977,Oehlert1992}. Under the structural model, this gives
\begin{equation}
\omega_{22}^{\Delta}
=
\frac{\lambda^2\sigma_U^2+\sigma_{\varepsilon_2}^2}{\alpha^2\sigma_G^2}.
\label{eq:omega22_delta_single_section5}
\end{equation}
A denominator-fixed Wald standard error, of the type commonly used in summary-data Wald-ratio implementations, keeps only the first term in \eqref{eq:wald_full_delta_section5} \cite{Hemani2018MRBase,TwoSampleMRPackage,TwoSampleMRWaldRatio}. In the present model, its probability target is
\begin{equation}
\omega_{22}^{\mathrm{pkg}}(\gamma)
=
\frac{(\gamma+\lambda)^2\sigma_U^2+\gamma^2\sigma_{\varepsilon_1}^2+
\sigma_{\varepsilon_2}^2}{\alpha^2\sigma_G^2}.
\label{eq:omega22_pkg_single_section5}
\end{equation}
Therefore
\begin{equation}
\omega_{22}^{\mathrm{pkg}}(\gamma)-\omega_{22}^{\Delta}
=
\frac{2\gamma\lambda\sigma_U^2+
\gamma^2(\sigma_U^2+\sigma_{\varepsilon_1}^2)}
{\alpha^2\sigma_G^2}.
\label{eq:omega22_single_difference_section5}
\end{equation}
At \(\gamma=0\), the two variance targets agree. Away from the null, they have no fixed ordering. If \(\lambda>0\) and \(\gamma<0\) with moderate magnitude, the package-style variance is smaller; if \(\gamma>0\), it is larger. The same sign-reversal mechanism appears for fixed-\(K\) IVW: the influence-function variance and the package-style fixed-effect summary variance agree at the null in the structural model, but can differ in opposite directions under positive and negative alternatives. The detailed fixed-\(K\) formulas are given in Appendix~\ref{app:ivw-package-variance}.

Figure~\ref{fig:omega22_diagnostics} displays the final \(\omega_{22}\) power diagnostics for the feasible statistic \(U_2\). The single-SNP run uses \(p=0.5\), \(\alpha=1\), \(\lambda=1\), \(\sigma_U^2=\sigma_{\varepsilon_1}^2=\sigma_{\varepsilon_2}^2=1\), \(n=5{,}000\), and \(m\in\{2{,}500,5{,}000,25{,}000\}\). The multi-SNP run uses \(K=10\), common variants with \(p=0.5\), equal effects \(\alpha=0.7\), \(\lambda=1.5\), \(\sigma_U^2=4\), \(\sigma_{\varepsilon_1}^2=10\), \(\sigma_{\varepsilon_2}^2=1\), \(n=5{,}000\), and the same external sample sizes. These are diagnostic settings rather than attempts to represent every MR application. They are nevertheless practically interpretable: the instruments are common and the multi-SNP first-stage strength is strong but plausible for genome-wide significant instruments. The difference between the two \(\omega_{22}\) targets is modest in these final runs, but the direction is systematic and reverses with the sign of \(\gamma\) relative to \(\lambda\). Thus package-style \(\omega_{22}\) should not be interpreted as uniformly conservative or uniformly anti-conservative away from the null.

\begin{figure}[htbp]
\centering
\begin{subfigure}{0.98\textwidth}
\centering
\safeincludegraphics[width=\textwidth]{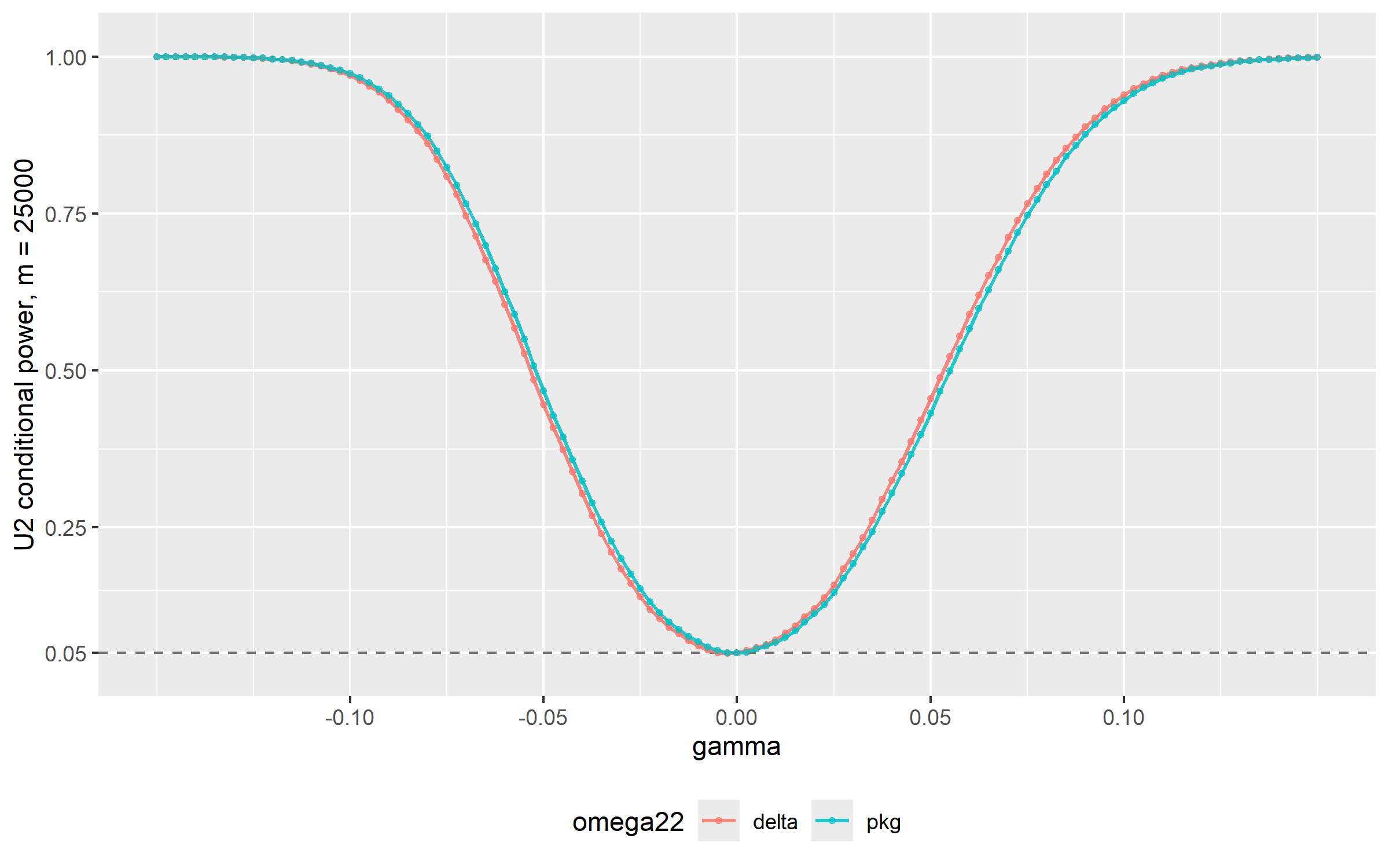}
\caption{Single-SNP U2.}
\end{subfigure}
\vspace{0.8em}
\begin{subfigure}{0.98\textwidth}
\centering
\safeincludegraphics[width=\textwidth]{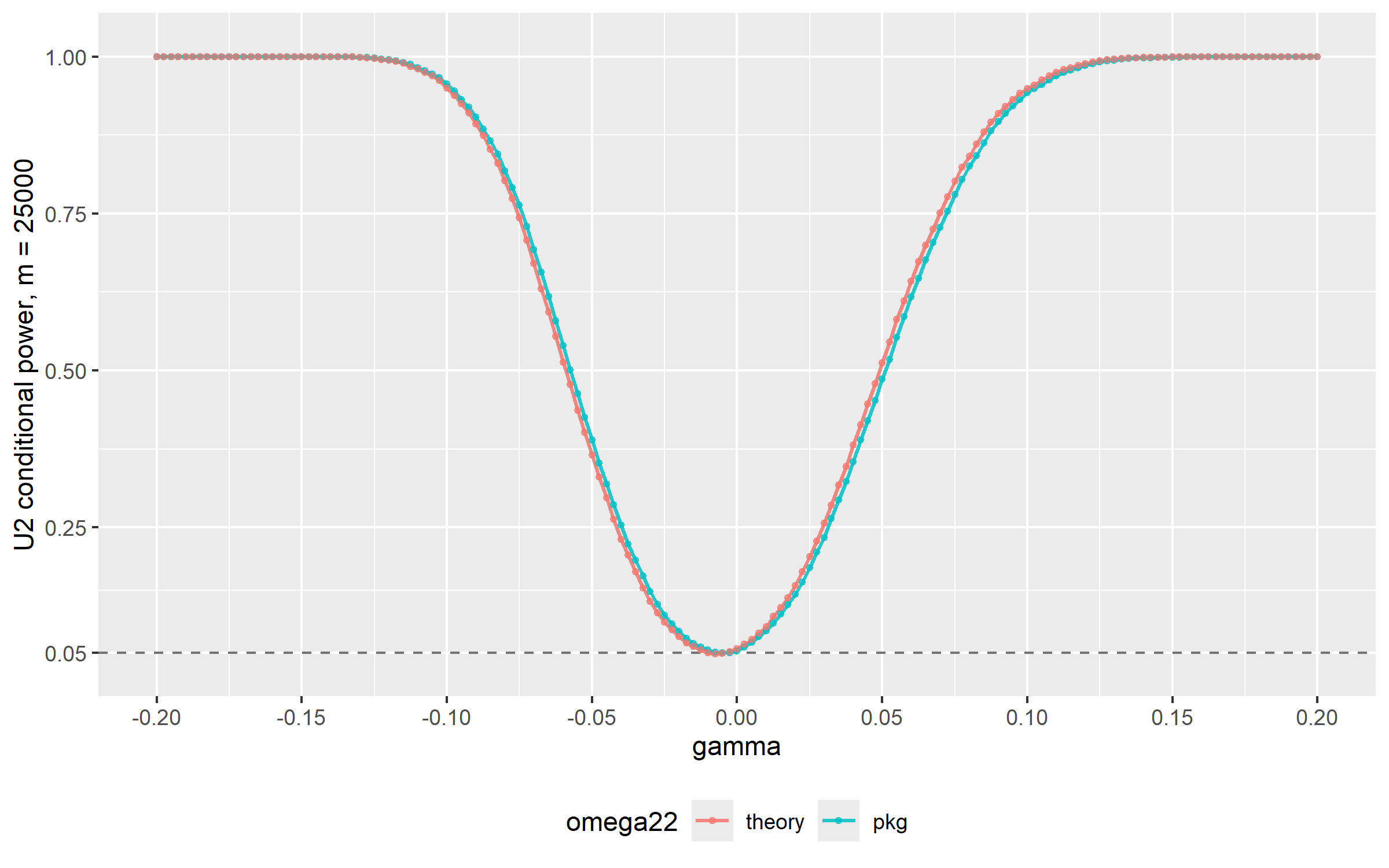}
\caption{Multi-SNP U2.}
\end{subfigure}
\caption{Effect of the \(\omega_{22}\) variance target on conditional power. The full first-order target and the package-style target agree at the null in these structural-model diagnostics, but differ away from the null. The sign of the difference reverses for negative versus positive causal effects when the confounding direction is fixed. This affects standardization and power, not the decorrelation coefficient.}
\label{fig:omega22_diagnostics}
\end{figure}

\subsection{Implementation recommendation}
\label{subsec:variance_recommendation}

The practical recommendation is therefore asymmetric. For \(\omega_{11}\), the robust influence-function variance is essential because \(\omega_{11}\) appears in the decorrelation coefficient itself. Replacing it by a classical homoskedastic target can leave first-order dependence between the screening statistic and the corrected numerator. For \(\omega_{22}\), the main issue is coherent standardization. A package-style Wald or IVW standard error may be reasonable for its original summary-data workflow, especially at the null where several targets coincide under the present structural model, but it generally estimates a different variance object away from the null.

Whenever individual-level data, or sufficient moment information, are available, the proposed implementation should estimate the full covariance matrix \(\Omega\) by the empirical influence-function covariance estimator in Section~\ref{sec:decorrelated}. This estimates \(\omega_{11}\), \(\omega_{12}\), and \(\omega_{22}\) as components of one joint first-order expansion. In applications where only summary-level package standard errors are available, the numerator decorrelation target \(\omega_{12}/\omega_{11}\) should be treated as the primary object needed for post-screening calibration, and the choice of \(\omega_{22}\) should be reported as a standardization and power sensitivity rather than as the mechanism that removes selection-induced type I error inflation.

\section{Discussion}
\label{sec:discussion}

This paper studies a source of invalid MR inference that arises from the workflow rather than from a violation of the usual instrumental-variable assumptions. In many exploratory analyses, MR is used only after an exposure--outcome pair has shown an observational association. The relevant null distribution is then the distribution of the MR statistic conditional on the association-screening event. The simulations, the bivariate Gaussian geometry, and the selected-density calculation all show the same phenomenon: a conventional MR statistic can be marginally calibrated and still be conditionally invalid after screening.

The proposed statistic targets this mechanism directly. The post-screening distortion is driven by the first-order covariance between the observational association estimator \(\hat\beta_{XY}\) and the MR estimator \(\hat\theta_n\). The feasible statistic \(U_2\) estimates and removes the component of \(\hat\theta_n\) that is linearly predictable from \(\hat\beta_{XY}\). Under the joint Gaussian first-order approximation, this decorrelation makes the corrected statistic asymptotically independent of screening events determined by \(\hat\beta_{XY}\). The procedure can therefore be understood as a selective-inference adjustment tailored to the MR follow-up workflow: validity is defined after the data-dependent decision to test has been made \cite{FithianSunTaylor2014,TaylorTibshirani2015,LeeSunSunTaylor2016}.

The external association estimate is central to the construction. It supplies an independent estimate of the population observational association \(\beta_{XY}\), allowing the statistic to remove the main-sample fluctuation that drives screening without reusing the same random quantity that defined the screening event. In applications, the required input is an independent exposure--outcome association estimate from a comparable population, possibly available as summary statistics from a previous association study. Standard SNP--exposure and SNP--outcome GWAS summary statistics alone do not provide this quantity unless they also include, or can be combined with, an independent estimate of the observational association between the same exposure and outcome.

The method is not intended to replace existing robustness tools for other MR problems. The theory assumes valid instruments under the structural model and focuses on post-screening calibration. Weak instruments, horizontal pleiotropy, linkage disequilibrium, many weak-instrument asymptotics, binary traits, nonlinear exposure models, and sample overlap between the main and external association studies are not addressed here. These issues are important but conceptually separate from the conditioning problem studied in this paper. In practice, post-screening calibration should be viewed as complementary to weak-instrument diagnostics and pleiotropy-robust MR methods.

The covariance target is also part of the methodology, not merely an implementation detail. In particular, \(\omega_{11}\) is the robust random-design variance of the observational OLS slope. Because \(\omega_{11}\) determines the decorrelation coefficient \(r=\omega_{12}/\omega_{11}\), using a classical homoskedastic slope variance can reintroduce the very dependence that the correction is designed to remove. Estimating \(\Omega\) through the empirical covariance of influence-function terms avoids this mismatch and keeps the numerator correction and denominator standardization tied to the same first-order expansion.

Several extensions are natural. One direction is to develop high-dimensional versions in which the number of genetic variants grows with sample size. Another is to combine the decorrelation idea with pleiotropy-robust estimators such as MR-Egger, weighted median, or mode-based procedures, each of which would require its own joint influence-function expansion with the screening statistic. A third direction is to model partial overlap between the main sample, the MR summary statistics, and the external exposure--outcome association study. Finally, binary outcomes and generalized exposure models would require replacing the OLS association and linear IVW influence functions by appropriate generalized-model analogues.

The main message is that a screen-then-MR workflow changes the inferential problem. When MR is used as a follow-up tool after observational association screening, ordinary MR \(p\)-values can overstate causal evidence because they are calibrated to the marginal null distribution rather than to the selected null distribution. A decorrelated statistic provides a direct way to calibrate causal testing to the screened workflow while retaining the familiar Wald-ratio and IVW estimators as starting points.

\appendix

\section{Supplement to Subsection~\ref{subsec:gaussian_geometry}: Selected Distribution after Association Screening}
\label{app:selected-density}

This appendix supplements the bivariate Gaussian geometry and selected-density calculation in Section~\ref{sec:problem}, especially Subsection~\ref{subsec:gaussian_geometry}. It derives the conditional density and the selected type I error formula used to explain why association screening can distort a marginally calibrated MR test.

Let
\[
W_1=\hat\beta_{XY},\qquad W_2=\hat\theta_n,
\]
where \(\hat\theta_n\) denotes the Wald ratio in the single-SNP case and the fixed-\(K\) IVW estimator in the multi-SNP case. Under the first-order Gaussian approximation established in Section~\ref{sec:joint},
\begin{equation}
\begin{pmatrix}W_1\\ W_2\end{pmatrix}
\approx
\Normal\left[
\begin{pmatrix}\beta_{XY}\\ \gamma\end{pmatrix},
\frac1n
\begin{pmatrix}
\omega_{11} & \omega_{12}\\
\omega_{12} & \omega_{22}
\end{pmatrix}
\right].
\label{eq:app-selected-joint-normal}
\end{equation}
Consider the two-sided association-screening event
\[
A_n=\{|W_1|>c_n\}.
\]
The conditional density of \(W_2\) given selection is obtained from the definition of conditional density:
\begin{equation}
 f_{W_2\mid A_n}(t)
 =
 \frac{\int_{|x|>c_n}f_{W_1,W_2}(x,t)\,dx}{\Prb(A_n)}.
\label{eq:app-selected-density-integral}
\end{equation}
Writing the joint density as \(f_{W_1\mid W_2}(x\mid t)f_{W_2}(t)\) gives the selection-weight representation
\begin{equation}
 f_{W_2\mid A_n}(t)
 =
 f_{W_2}(t)
 \frac{\Prb(|W_1|>c_n\mid W_2=t)}{\Prb(A_n)}.
\label{eq:app-selected-density-weight}
\end{equation}
Thus the selected density is the marginal density multiplied by the probability of passing the association screen at the realized value of the MR estimator.

By the bivariate normal conditioning formula applied to \eqref{eq:app-selected-joint-normal},
\begin{equation}
W_1\mid W_2=t
\approx
\Normal\left[
\beta_{XY}+\frac{\omega_{12}}{\omega_{22}}(t-\gamma),
\frac1n\left(\omega_{11}-\frac{\omega_{12}^2}{\omega_{22}}\right)
\right].
\label{eq:app-W1-given-W2}
\end{equation}
Consequently,
\begin{align}
\Prb(|W_1|>c_n\mid W_2=t)
&=
\Phi\left[
\frac{-c_n-\beta_{XY}-\omega_{12}(t-\gamma)/\omega_{22}}
{\{n^{-1}(\omega_{11}-\omega_{12}^2/\omega_{22})\}^{1/2}}
\right]
\notag\\
&\quad+
\Phi\left[
\frac{\beta_{XY}+\omega_{12}(t-\gamma)/\omega_{22}-c_n}
{\{n^{-1}(\omega_{11}-\omega_{12}^2/\omega_{22})\}^{1/2}}
\right],
\label{eq:app-selection-weight-W}
\end{align}
and
\begin{equation}
\Prb(A_n)=
\Phi\left(\frac{-c_n-\beta_{XY}}{\sqrt{\omega_{11}/n}}\right)
+
\Phi\left(\frac{\beta_{XY}-c_n}{\sqrt{\omega_{11}/n}}\right).
\label{eq:app-selection-prob-W}
\end{equation}
Equations~\eqref{eq:app-selected-density-weight}--\eqref{eq:app-selection-prob-W} give the selected density of the unstandardized MR estimator.

For the standardized form used in Section~\ref{sec:problem}, define
\[
Z_1=\frac{\sqrt n(W_1-\beta_{XY})}{\sqrt{\omega_{11}}},
\qquad
Z_2=\frac{\sqrt n(W_2-\gamma)}{\sqrt{\omega_{22}}},
\]
\[
\mu_n=\frac{\sqrt n\beta_{XY}}{\sqrt{\omega_{11}}},
\qquad
c=\frac{\sqrt n c_n}{\sqrt{\omega_{11}}},
\qquad
\rho=\frac{\omega_{12}}{\sqrt{\omega_{11}\omega_{22}}}.
\]
Then \((Z_1,Z_2)\) is approximately bivariate standard normal with correlation \(\rho\), and
\[
A_n=\{|Z_1+\mu_n|>c\}.
\]
Since
\[
Z_1\mid Z_2=z\sim \Normal(\rho z,1-\rho^2),
\]
the conditional probability of passing the screen is
\begin{equation}
s_{\mu,c,\rho}(z)
=
\Phi\left(\frac{-c-\mu-\rho z}{\sqrt{1-\rho^2}}\right)
+
\Phi\left(\frac{\mu+\rho z-c}{\sqrt{1-\rho^2}}\right),
\label{eq:app-selection-weight-standardized}
\end{equation}
and the marginal selection probability is
\begin{equation}
p_{\mu,c}
=
\Phi(-c-\mu)+\Phi(\mu-c).
\label{eq:app-selection-prob-standardized}
\end{equation}
Therefore,
\begin{equation}
f_{Z_2\mid A_n}(z)
=
\phi(z)\frac{s_{\mu_n,c,\rho}(z)}{p_{\mu_n,c}}.
\label{eq:app-selected-z-density}
\end{equation}
For a two-sided level-\(a\) MR test based on the uncorrected standard normal reference distribution,
\begin{equation}
\alpha_{\mathrm{sel}}(a)
=
\Prb(|Z_2|>z_{1-a/2}\mid A_n)
=
\int_{|z|>z_{1-a/2}}
\phi(z)\frac{s_{\mu_n,c,\rho}(z)}{p_{\mu_n,c}}\,dz.
\label{eq:app-selected-alpha}
\end{equation}
If \(\rho=0\), then \(s_{\mu,c,0}(z)=p_{\mu,c}\) for all \(z\), so the selected density in \eqref{eq:app-selected-z-density} is \(\phi(z)\) and the selected type I error in \eqref{eq:app-selected-alpha} equals the nominal level. If \(\rho\ne0\), the selection weight depends on \(z\), and the selected distribution of the MR statistic is tilted relative to its marginal distribution.

\section{Proofs of Lemma~\ref{lem:joint-linearization} and Theorems~\ref{thm:single}--\ref{thm:multi}: Joint Asymptotic Distributions}
\label{app:joint-asymptotics}

This appendix proves the joint asymptotic results stated in Section~\ref{sec:joint}. The proof uses only finite-dimensional empirical-moment theory: sample covariances are asymptotically linear averages of product variables, the multivariate central limit theorem gives a Gaussian limit for those averages, and the delta method transfers the limit to the covariance-ratio estimators. No normality assumption on \(G\), \(U\), \(\varepsilon_1\), or \(\varepsilon_2\) is used.

Under the zero-mean normalization of Section~\ref{sec:joint}, write
\[
\delta=Y-\gamma X=\lambda U+\varepsilon_2,
\qquad
\tau^2=\Var(\delta)=\lambda^2\sigma_U^2+\sigma_{\varepsilon_2}^2,
\]
\[
V_X=\Var(X),
\qquad
\kappa=\frac{\lambda\sigma_U^2}{V_X},
\qquad
\beta_{XY}=\gamma+\kappa,
\qquad
e=Y-\beta_{XY}X=\delta-\kappa X.
\]
The sample covariance satisfies, for any two variables \(A\) and \(B\) with finite second moments,
\begin{equation}
S_{AB}-\Cov(A,B)
=
\frac1n\sum_{i=1}^n\{A_iB_i-\E(AB)\}+o_p(n^{-1/2}).
\label{eq:app-sample-cov-linear}
\end{equation}
Indeed, under zero means, \(S_{AB}=n^{-1}\sum_i A_iB_i-\bar A\bar B\), and \(\bar A\bar B=O_p(n^{-1})\). In the nonzero-mean case, the same display holds with all variables centered at their population means. The finite-fourth-moment assumption ensures that all products appearing below have finite variances.

\subsection{Proof of Lemma~\ref{lem:joint-linearization}}
\label{app:linearization-proof}

For the observational association estimator, let \(A_n=S_{XY}\), \(B_n=S_{XX}\), \(A=\Cov(X,Y)\), and \(B=V_X\). The map \(f(a,b)=a/b\) has derivative \((1/B,-A/B^2)\) at \((A,B)\). Hence
\[
\hat\beta_{XY}-\beta_{XY}
=
\frac{1}{V_X}(S_{XY}-\Cov(X,Y))
-
\frac{\beta_{XY}}{V_X}(S_{XX}-V_X)
+o_p(n^{-1/2}).
\]
Using \eqref{eq:app-sample-cov-linear},
\[
\sqrt n(\hat\beta_{XY}-\beta_{XY})
=
\frac1{\sqrt n}\sum_{i=1}^n
\frac{X_iY_i-\beta_{XY}X_i^2}{V_X}+o_p(1)
=
\frac1{\sqrt n}\sum_{i=1}^n\phi_{1i}+o_p(1),
\]
where
\[
\phi_1=\frac{X(Y-\beta_{XY}X)}{V_X}=\frac{Xe}{V_X}.
\]
The expectation of \(\phi_1\) is zero because \(\E\{X(Y-\beta_{XY}X)\}=0\) by the population linear-projection definition of \(\beta_{XY}\).

In the single-SNP case,
\[
\hat\beta_{MR}=\frac{S_{GY}}{S_{GX}},
\qquad
\gamma=\frac{\Cov(G,Y)}{\Cov(G,X)}.
\]
Since \(\Cov(G,X)=\alpha\sigma_G^2\ne0\), the same ratio expansion gives
\[
\sqrt n(\hat\beta_{MR}-\gamma)
=
\frac1{\sqrt n}\sum_{i=1}^n
\frac{G_iY_i-\gamma G_iX_i}{\Cov(G,X)}+o_p(1)
=
\frac1{\sqrt n}\sum_{i=1}^n\phi^S_{2i}+o_p(1),
\]
where
\[
\phi_2^S=\frac{G(Y-\gamma X)}{\Cov(G,X)}=\frac{G\delta}{\alpha\sigma_G^2}.
\]

In the fixed-\(K\) multi-SNP case, note that
\[
\hat\beta_{Yj}-\gamma\hat\beta_{Xj}
=
\frac{S_{G_j(Y-\gamma X)}}{S_{G_jG_j}}
=
\frac{S_{G_j\delta}}{S_{G_jG_j}}.
\]
Therefore
\begin{align*}
\hat\beta_{IVW}-\gamma
&=
\frac{\sum_{j=1}^K\pi_{j,n}\hat\beta_{Xj}(\hat\beta_{Yj}-\gamma\hat\beta_{Xj})}
{\sum_{j=1}^K\pi_{j,n}\hat\beta_{Xj}^2} \\
&=
\frac{\sum_{j=1}^K\pi_{j,n}\hat\beta_{Xj}S_{G_j\delta}/S_{G_jG_j}}
{\sum_{j=1}^K\pi_{j,n}\hat\beta_{Xj}^2}.
\end{align*}
By the law of large numbers, \(\hat\beta_{Xj}\toP\alpha_j\), \(S_{G_jG_j}\toP\sigma_{G_j}^2\), and \(\pi_{j,n}\toP\pi_j\). Since \(K\) is fixed and \(S_{G_j\delta}=O_p(n^{-1/2})\), Slutsky's theorem yields
\[
\hat\beta_{IVW}-\gamma
=
\sum_{j=1}^K q_jS_{G_j\delta}+o_p(n^{-1/2}),
\qquad
q_j=
\frac{\pi_j\alpha_j/\sigma_{G_j}^2}{\sum_{\ell=1}^K\pi_\ell\alpha_\ell^2}.
\]
Using \eqref{eq:app-sample-cov-linear} again,
\[
\sqrt n(\hat\beta_{IVW}-\gamma)
=
\frac1{\sqrt n}\sum_{i=1}^n
\left(\sum_{j=1}^Kq_jG_{ji}\right)\delta_i+o_p(1).
\]
Thus the multi-SNP influence function is
\[
\phi_2^M=
\left(\sum_{j=1}^Kq_jG_j\right)(Y-\gamma X),
\]
which is the same expression as \eqref{eq:phi2_multi_compact} in Section~\ref{sec:joint}. Combining the first component with either \(\phi_2^S\) or \(\phi_2^M\) proves the asymptotically linear representation in Lemma~\ref{lem:joint-linearization}. The multivariate central limit theorem applied to \((\phi_1,\phi_2)\) gives the joint Gaussian limit with covariance matrix \(\Omega=\Var(\phi_1,\phi_2)^\top\).

\subsection{Proof of Theorem~\ref{thm:single}: single-SNP case}
\label{app:single-proof}

The single-SNP influence functions are
\[
\phi_1=\frac{X(\delta-\kappa X)}{V_X},
\qquad
\phi_2^S=\frac{G\delta}{\alpha\sigma_G^2}.
\]
The joint central limit theorem in Lemma~\ref{lem:joint-linearization} gives the convergence in \eqref{eq:single_joint_limit}; it remains only to compute the covariance entries.

For the first entry,
\[
\omega_{11}=\Var(\phi_1)=\frac{\E\{X^2(Y-\beta_{XY}X)^2\}}{V_X^2},
\]
because \(\E\{X(Y-\beta_{XY}X)\}=0\). This is the robust random-design OLS slope variance.

For the second entry,
\[
\omega_{22}^S
=\frac{\Var(G\delta)}{\alpha^2\sigma_G^4}.
\]
Since \(G\) is independent of \(\delta=\lambda U+\varepsilon_2\), both variables have mean zero, and \(\Var(\delta)=\tau^2\),
\[
\Var(G\delta)=\sigma_G^2\tau^2.
\]
Thus
\[
\omega_{22}^S=\frac{\tau^2}{\alpha^2\sigma_G^2}
=
\frac{\lambda^2\sigma_U^2+\sigma_{\varepsilon_2}^2}{\alpha^2\sigma_G^2}.
\]

Finally,
\begin{align*}
\omega_{12}
&=\Cov(\phi_1,\phi_2^S) \\
&=\frac1{V_X\alpha\sigma_G^2}
\left\{\E(XG\delta^2)-\kappa\E(X^2G\delta)\right\}.
\end{align*}
Because \(X=\alpha G+U+\varepsilon_1\) and \(G\) is independent of \((U,\varepsilon_1,\delta)\),
\[
\E(XG\delta^2)=\alpha\sigma_G^2\E(\delta^2)=\alpha\sigma_G^2\tau^2.
\]
Moreover,
\[
\E(X^2G\delta)=2\alpha\sigma_G^2\E\{(U+\varepsilon_1)\delta\}
=2\alpha\sigma_G^2\lambda\sigma_U^2.
\]
The terms involving \(\E(G^3)\) vanish because \(\E(\delta)=0\), and the terms involving \(\E(G)\) vanish because \(\E(G)=0\). Therefore
\[
\omega_{12}
=
\frac{\tau^2}{V_X}-\frac{2\kappa\lambda\sigma_U^2}{V_X}
=
\frac{\lambda^2\sigma_U^2+\sigma_{\varepsilon_2}^2}{V_X}
-
\frac{2\lambda^2\sigma_U^4}{V_X^2}.
\]
This proves Theorem~\ref{thm:single}.

\subsection{Proof of Theorem~\ref{thm:multi}: fixed-\texorpdfstring{\(K\)}{K} IVW case}
\label{app:multi-proof}

The multi-SNP influence functions are
\[
\phi_1=\frac{X(\delta-\kappa X)}{V_X},
\qquad
\phi_2^M=T_G\delta,
\qquad
T_G=\sum_{j=1}^K q_jG_j,
\]
where
\[
q_j=\frac{\pi_j\alpha_j/\sigma_{G_j}^2}{D_\pi},
\qquad
D_\pi=\sum_{\ell=1}^K\pi_\ell\alpha_\ell^2.
\]
Lemma~\ref{lem:joint-linearization} gives the joint central limit theorem. The first variance entry is again
\[
\omega_{11}=\frac{\E\{X^2(Y-\beta_{XY}X)^2\}}{V_X^2}.
\]

For \(\omega_{22}^M\), \(T_G\) depends only on the instruments and is independent of \(\delta\). Since \(\E(T_G)=0\) and \(\E(\delta)=0\),
\[
\omega_{22}^M=\Var(T_G\delta)=\Var(T_G)\Var(\delta).
\]
Pairwise uncorrelated SNPs give
\[
\Var(T_G)=\sum_{j=1}^K q_j^2\sigma_{G_j}^2
=
\frac{\sum_{j=1}^K\pi_j^2\alpha_j^2/\sigma_{G_j}^2}
{\left(\sum_{\ell=1}^K\pi_\ell\alpha_\ell^2\right)^2}.
\]
Thus
\[
\omega_{22}^M
=
(\lambda^2\sigma_U^2+\sigma_{\varepsilon_2}^2)
\frac{\sum_{j=1}^K\pi_j^2\alpha_j^2/\sigma_{G_j}^2}
{\left(\sum_{\ell=1}^K\pi_\ell\alpha_\ell^2\right)^2}.
\]

For the cross-covariance, write \(A=\sum_{j=1}^K\alpha_jG_j\), so that \(X=A+U+\varepsilon_1\). Since
\[
\sum_{j=1}^K \alpha_jq_j\sigma_{G_j}^2
=
\frac{\sum_{j=1}^K\pi_j\alpha_j^2}{D_\pi}=1,
\]
we have
\[
\E(XT_G\delta^2)=\E(AT_G)\E(\delta^2)=\tau^2.
\]
Similarly,
\[
\E(X^2T_G\delta)=2\E(AT_G)\E\{(U+\varepsilon_1)\delta\}=2\lambda\sigma_U^2.
\]
It follows that
\[
\omega_{12}
=\frac1{V_X}\{\tau^2-2\kappa\lambda\sigma_U^2\}
=
\frac{\lambda^2\sigma_U^2+\sigma_{\varepsilon_2}^2}{V_X}
-
\frac{2\lambda^2\sigma_U^4}{V_X^2}.
\]
This proves Theorem~\ref{thm:multi}.

\section{Proof of Proposition~\ref{prop:plugin_consistency}}
\label{app:omega-estimation}

This appendix proves Proposition~\ref{prop:plugin_consistency}. It supplements Subsection~\ref{subsec:plugin_estimation}, where the feasible estimator \(\hat\Omega\) in \eqref{eq:omega_hat_if} is introduced.

Let \(\phi_i=(\phi_{1i},\phi_{2i})^\top\) denote the population influence-function vector in the relevant single-SNP or multi-SNP setting, and let \(\hat\phi_i=(\hat\phi_{1i},\hat\phi_{2i})^\top\) be the sample analogue defined in \eqref{eq:phi1_hat}--\eqref{eq:phi2_multi_hat}. The proof uses the following elementary lemma.

\begin{lemma}[Empirical covariance with estimated influence functions]
\label{lem:estimated-if-cov}
Suppose \(\E\|\phi_i\|^2<\infty\) and
\[
\frac1n\sum_{i=1}^n\|\hat\phi_i-\phi_i\|^2=o_p(1).
\]
Then
\[
\frac1n\sum_{i=1}^n(\hat\phi_i-\bar{\hat\phi})(\hat\phi_i-\bar{\hat\phi})^\top
\toP
\E(\phi_i\phi_i^\top)=\Omega.
\]
\end{lemma}

\begin{proof}
First,
\[
\begin{aligned}
&\left\|
\frac1n\sum_{i=1}^n\hat\phi_i\hat\phi_i^\top-
\frac1n\sum_{i=1}^n\phi_i\phi_i^\top
\right\| \\
&\quad\le
\left(\frac1n\sum_{i=1}^n\|\hat\phi_i-\phi_i\|^2\right)^{1/2}
\left(\frac1n\sum_{i=1}^n\|\hat\phi_i\|^2\right)^{1/2} \\
&\qquad+
\left(\frac1n\sum_{i=1}^n\|\phi_i\|^2\right)^{1/2}
\left(\frac1n\sum_{i=1}^n\|\hat\phi_i-\phi_i\|^2\right)^{1/2}.
\end{aligned}
\]
The first factor involving \(\hat\phi_i-\phi_i\) is \(o_p(1)\); the empirical second moments of \(\phi_i\) are \(O_p(1)\) by the law of large numbers, and those of \(\hat\phi_i\) are also \(O_p(1)\) by the displayed mean-square convergence. Hence the difference of uncentered empirical second moments is \(o_p(1)\). The law of large numbers gives
\[
\frac1n\sum_{i=1}^n\phi_i\phi_i^\top\toP \E(\phi_i\phi_i^\top)=\Omega.
\]
Finally, \(\bar{\hat\phi}=n^{-1}\sum_i\hat\phi_i=o_p(1)\), because \(n^{-1}\sum_i\phi_i=o_p(1)\) and \(n^{-1}\sum_i(\hat\phi_i-\phi_i)=o_p(1)\) by Cauchy's inequality. Empirical centering therefore changes the covariance estimator by \(\bar{\hat\phi}\bar{\hat\phi}^{\top}=o_p(1)\). This proves the lemma.
\end{proof}

It remains to verify the mean-square convergence condition. In the single-SNP case, the probability limits
\[
S_{XX}\toP V_X,
\qquad
S_{GX}\toP \Cov(G,X),
\qquad
\hat\beta_{XY}\toP\beta_{XY},
\qquad
\hat\beta_{MR}\toP\gamma
\]
follow from the law of large numbers, relevance, and the continuous mapping theorem. The finite-fourth-moment condition implies that the empirical second moments of the relevant products are tight. Therefore the sample analogues \eqref{eq:phi1_hat} and \eqref{eq:phi2_single_hat} converge in empirical mean square to \(\phi_1\) and \(\phi_2^S\). Lemma~\ref{lem:estimated-if-cov} then gives \(\hat\Omega\toP\Omega\).

In the multi-SNP case, the same argument applies to \(\hat\phi_{1i}\). For the second component, \(K\) is fixed and
\[
\hat\beta_{Xj}\toP\alpha_j,
\qquad
S_{G_jG_j}\toP\sigma_{G_j}^2,
\qquad
\pi_{j,n}\toP\pi_j,
\qquad j=1,\ldots,K.
\]
Thus \(\hat q_j\toP q_j\) for each \(j\), and \(\hat\beta_{IVW}\toP\gamma\). The finite-fourth-moment condition again gives empirical mean-square convergence of \(\hat\phi_{2i}\) to \(\phi_2^M\). Lemma~\ref{lem:estimated-if-cov} yields \(\hat\Omega\toP\Omega\). Since \(\omega_{11}>0\) and \(V_1>0\), the continuous mapping theorem gives
\[
\hat r=\frac{\hat\omega_{12}}{\hat\omega_{11}}\toP r,
\qquad
\hat V_1
=
\hat\omega_{22}+\left(\frac nm-1\right)\frac{\hat\omega_{12}^2}{\hat\omega_{11}}
\toP V_1,
\]
provided \(n/m\to\eta\in(0,\infty)\). This proves Proposition~\ref{prop:plugin_consistency}.

\section{Proofs of Theorems~\ref{thm:U0}, \ref{thm:U1}, and~\ref{thm:U2}: Decorrelated Statistics}
\label{app:decorrelated}

This appendix proves the three decorrelation results in Section~\ref{sec:decorrelated}. The proof is written in the unified notation \(\hat\theta_n=\hat\beta_{MR}\) for one SNP and \(\hat\theta_n=\hat\beta_{IVW}\) for fixed \(K\) SNPs. By Section~\ref{sec:joint},
\begin{equation}
\sqrt n
\begin{pmatrix}
\hat\beta_{XY}-\beta_{XY}\\
\hat\theta_n-\gamma
\end{pmatrix}
\toD
\Normal\left\{
0,
\begin{pmatrix}
\omega_{11} & \omega_{12}\\
\omega_{12} & \omega_{22}
\end{pmatrix}
\right\}.
\label{eq:app-decor-joint-raw}
\end{equation}
where the symbol \(\Normal\{0,\Omega\}\) in \eqref{eq:app-decor-joint-raw} means a centered bivariate normal distribution with covariance matrix \(\Omega\).  Assume \(\Omega\) is positive definite and define
\[
r=\frac{\omega_{12}}{\omega_{11}},
\qquad
V_0=\omega_{22}-\frac{\omega_{12}^2}{\omega_{11}}>0.
\]
For readability, set
\[
A_n^*=\sqrt n(\hat\beta_{XY}-\beta_{XY}),
\qquad
B_n=\sqrt n(\hat\theta_n-\gamma).
\]

\subsection{Proof of Theorem~\ref{thm:U0}: oracle statistic}
\label{app:proof-U0}

The oracle numerator has the centered form \(B_n-rA_n^*\). By the continuous mapping theorem,
\[
\begin{pmatrix}A_n^*\\ B_n-rA_n^*\end{pmatrix}
\toD
\Normal\left[
0,
\begin{pmatrix}
\omega_{11} & \omega_{12}-r\omega_{11}\\
\omega_{12}-r\omega_{11} & \omega_{22}-2r\omega_{12}+r^2\omega_{11}
\end{pmatrix}
\right].
\]
Because \(r=\omega_{12}/\omega_{11}\), the off-diagonal entry is zero and the second diagonal entry is \(V_0\). Hence \(B_n-rA_n^*\toD\Normal(0,V_0)\) under the null, and
\[
U_0=\frac{B_n-rA_n^*}{\sqrt{V_0}}+o_p(1)\toD\Normal(0,1).
\]
Under the local alternative \(\gamma=\gamma_0+\Delta/\sqrt n\), the numerator equals
\[
\sqrt n\{\hat\theta_n-\gamma-r(\hat\beta_{XY}-\beta_{XY})\}+\Delta,
\]
so the limit is \(\Normal(\Delta/\sqrt{V_0},1)\). The post-selection statement follows from the last subsection of this appendix. This proves Theorem~\ref{thm:U0}.

\subsection{Proof of Theorem~\ref{thm:U1}: external-association statistic}
\label{app:proof-U1}

Let \(\hat\beta_{XY}^{(s)}\) be independent of the main sample and satisfy
\[
\sqrt m(\hat\beta_{XY}^{(s)}-\beta_{XY})\toD\Normal(0,\omega_{11}),
\qquad
n/m\to\eta\in(0,\infty).
\]
Let \(C_m=\sqrt m(\hat\beta_{XY}^{(s)}-\beta_{XY})\). Then \(C_m\) is asymptotically independent of \((A_n^*,B_n)\). The centered numerator of \(U_1\), multiplied by \(\sqrt n\), is
\[
B_n-rA_n^*+r\sqrt{\frac{n}{m}}C_m.
\]
The first two terms have variance \(V_0\), the external term has variance \(r^2\eta\omega_{11}\), and the external term is independent of the main-sample terms. Therefore the limiting variance is
\[
V_0+r^2\eta\omega_{11}
=
\omega_{22}-\frac{\omega_{12}^2}{\omega_{11}}
+
\eta\frac{\omega_{12}^2}{\omega_{11}}
=
\omega_{22}+(\eta-1)\frac{\omega_{12}^2}{\omega_{11}},
\]
which is the limit of \(V_1\). Hence \(U_1\toD\Normal(0,1)\) under \(H_0^\gamma\), and under \(\gamma=\gamma_0+\Delta/\sqrt n\) the mean shifts by \(\Delta/\sqrt{V_1}\). Moreover, the covariance between the limiting association coordinate and the limiting numerator is
\[
\Cov\{A,B-rA+r\sqrt\eta C\}
=\omega_{12}-r\omega_{11}=0,
\]
because the external limit \(C\) is independent of \(A\). The post-selection statement follows from the last subsection. This proves Theorem~\ref{thm:U1}.

\subsection{Proof of Theorem~\ref{thm:U2}: feasible statistic}
\label{app:proof-U2}

By Proposition~\ref{prop:plugin_consistency}, \(\hat r\toP r\) and \(\hat V_1\toP V_1\). Let
\[
D_n=\hat\beta_{XY}-\hat\beta_{XY}^{(s)}.
\]
Since \(\hat\beta_{XY}-\beta_{XY}=O_p(n^{-1/2})\), \(\hat\beta_{XY}^{(s)}-\beta_{XY}=O_p(m^{-1/2})\), and \(n/m\to\eta\in(0,\infty)\), we have \(D_n=O_p(n^{-1/2})\). The difference between the feasible and external-association numerators is
\[
-(\hat r-r)D_n.
\]
After multiplication by \(\sqrt n\), this difference is
\[
-(\hat r-r)\sqrt nD_n=o_p(1)O_p(1)=o_p(1).
\]
The denominator ratio satisfies \(\sqrt{\hat V_1/V_1}\toP1\). Thus
\[
U_2-U_1=o_p(1).
\]
Slutsky's theorem gives the null and local-alternative limits stated in Theorem~\ref{thm:U2}. The post-selection statement follows from the next subsection.

\subsection{Conditional post-selection implication}
\label{app:post-selection-implication}

Theorems~\ref{thm:U0}--\ref{thm:U2} condition on screening events determined by the association statistic. The precise argument is as follows. Let
\[
S_n=\frac{\sqrt n(\hat\beta_{XY}-\beta_{XY})}{\sqrt{\omega_{11}}}.
\]
For \(j=0,1,2\), the preceding proofs imply
\[
(S_n,U_j)\toD (S,U),
\qquad
(S,U)\sim\Normal\left(
\begin{pmatrix}0\\0\end{pmatrix},
\begin{pmatrix}1&0\\0&1\end{pmatrix}
\right)
\]
under the null, with the usual interpretation that \(U_2-U_1=o_p(1)\). Let \(A_n\) be a screening event whose limiting form is \(\{S_n\in\mathcal A\}\), where \(\Prb(S\in\partial\mathcal A)=0\) and \(\Prb(S\in\mathcal A)>0\). Then, for any continuity set \(B\) of the standard normal law,
\[
\Prb(U_j\in B\mid A_n)
=
\frac{\Prb(U_j\in B,A_n)}{\Prb(A_n)}
\longrightarrow
\frac{\Prb(U\in B,S\in\mathcal A)}{\Prb(S\in\mathcal A)}.
\]
Since \(S\) and \(U\) are independent in the limiting Gaussian experiment, the last ratio equals \(\Prb(U\in B)\). Therefore
\[
\mathcal L(U_j\mid A_n)\Rightarrow \Normal(0,1),
\qquad j=0,1,2,
\]
which proves the post-screening validity statements.

\section{Supplement to Subsection~\ref{subsec:omega11_interpretation}: Heteroskedasticity of the Observational Regression and the Role of \texorpdfstring{\(\omega_{11}\)}{omega11}}
\label{app:heteroskedasticity}

This appendix supplements Subsection~\ref{subsec:omega11_interpretation}. It gives the derivations behind the robust variance target for \(\omega_{11}\), its difference from the classical homoskedastic variance, and the allele-frequency calculation used in the \(\omega_{11}\) diagnostic simulations.

Let
\[
X_c=X-\E(X),
\qquad
Y_c=Y-\E(Y),
\qquad
V_X=\Var(X),
\qquad
\beta_{XY}=\frac{\Cov(X,Y)}{V_X},
\]
and define the population linear-projection residual
\[
e=Y_c-\beta_{XY}X_c.
\]
The population OLS orthogonality condition is \(\E(X_ce)=0\). The random-design OLS expansion gives
\[
\sqrt n(\hat\beta_{XY}-\beta_{XY})
=
\frac1{\sqrt n}\sum_{i=1}^n\frac{X_{ci}e_i}{V_X}+o_p(1),
\]
so
\begin{equation}
\omega_{11}=\frac{\E(X_c^2e^2)}{V_X^2}.
\label{eq:app-omega11-robust}
\end{equation}
The classical homoskedastic target is
\begin{equation}
\omega_{11}^{\mathrm{cls}}=\frac{\Var(e)}{V_X}.
\label{eq:app-omega11-classical}
\end{equation}
Since \(\E(X_c^2)=V_X\), the exact difference is
\begin{equation}
\omega_{11}-\omega_{11}^{\mathrm{cls}}
=
\frac{\E(X_c^2e^2)-V_X\E(e^2)}{V_X^2}
=
\frac{\Cov(X_c^2,e^2)}{V_X^2}.
\label{eq:app-omega11-difference-cov}
\end{equation}
Thus the classical variance is appropriate only when the squared projection residual is uncorrelated with the squared exposure.

Under the single-SNP structural model, write
\[
X_c=\alpha G_c+U_c+\varepsilon_1,
\qquad
\delta=Y_c-\gamma X_c=\lambda U_c+\varepsilon_2,
\qquad
\kappa=\frac{\lambda\sigma_U^2}{V_X}.
\]
Then \(\beta_{XY}=\gamma+\kappa\) and
\[
e=\delta-\kappa X_c
=-\kappa\alpha G_c+(\lambda-\kappa)U_c-\kappa\varepsilon_1+\varepsilon_2.
\]
For independent centered components \(Z_l\), the fourth-moment identity
\[
\Cov\left\{\left(\sum_l a_lZ_l\right)^2,
\left(\sum_l b_lZ_l\right)^2\right\}
=
2\left(\sum_l a_lb_l\Var(Z_l)\right)^2
+
\sum_l a_l^2b_l^2\kappa_l^{(4)}
\]
holds, where \(\kappa_l^{(4)}=\E(Z_l^4)-3\Var(Z_l)^2\). Here the first term is zero because \(\Cov(X_c,e)=0\). If \(\varepsilon_1\) and \(\varepsilon_2\) are normal, their fourth cumulants vanish, and
\begin{equation}
\Cov(X_c^2,e^2)
=
\alpha^4\kappa^2\kappa_G^{(4)}+(\lambda-\kappa)^2\kappa_U^{(4)}.
\label{eq:app-omega11-cumulant-single}
\end{equation}
Consequently,
\begin{equation}
\omega_{11}
=
\omega_{11}^{\mathrm{cls}}
+
\frac{\alpha^4\kappa^2\kappa_G^{(4)}+(\lambda-\kappa)^2\kappa_U^{(4)}}{V_X^2}.
\label{eq:app-omega11-cumulant-decomp}
\end{equation}
When \(G\) and \(U\) satisfy Gaussian-type fourth-moment identities, \(\kappa_G^{(4)}=\kappa_U^{(4)}=0\), the robust and classical targets coincide. Otherwise they can differ substantially.

For the allele-frequency calculation used in Section~\ref{subsec:omega11_interpretation}, assume \(G\sim\mathrm{Bin}(2,p)\), \(U\) is normal, and set \(q=p(1-p)\). Then \(\Var(G)=2q\), \(\kappa_U^{(4)}=0\), and
\[
\kappa_G^{(4)}=\E\{(G-2p)^4\}-3\{2q\}^2=2q(1-6q).
\]
Also
\[
V_X=2\alpha^2q+\sigma_U^2+\sigma_{\varepsilon_1}^2,
\qquad
\omega_{11}^{\mathrm{cls}}
=
\frac{\tau^2}{V_X}-\frac{\lambda^2\sigma_U^4}{V_X^2},
\qquad
\tau^2=\lambda^2\sigma_U^2+\sigma_{\varepsilon_2}^2.
\]
Therefore
\begin{equation}
R_{\omega_{11}}(p)
=
\frac{\omega_{11}-\omega_{11}^{\mathrm{cls}}}{\omega_{11}^{\mathrm{cls}}}
=
\frac{2\alpha^4\lambda^2\sigma_U^4q(1-6q)}
{V_X^2\{\tau^2V_X-\lambda^2\sigma_U^4\}}.
\label{eq:app-omega11-relative-diff}
\end{equation}
The denominator is positive because
\[
\tau^2V_X-\lambda^2\sigma_U^4
=
\lambda^2\sigma_U^2(2\alpha^2q+\sigma_{\varepsilon_1}^2)
+\sigma_{\varepsilon_2}^2V_X>0.
\]
Thus the sign of \(R_{\omega_{11}}(p)\) is the sign of \(q(1-6q)\). On \(p\in(0,1/2]\), the only nondegenerate zero solves \(p(1-p)=1/6\), namely
\[
p_0=\frac{3-\sqrt3}{6}\approx0.2113.
\]
The robust variance is larger than the classical variance for \(0<p<p_0\), smaller for \(p_0<p\le1/2\), and equal at \(p=p_0\).

For independent multi-SNP instruments with \(G_j\sim\mathrm{Bin}(2,p)\), define
\[
A=\sum_{j=1}^K\alpha_j(G_j-\E G_j),
\qquad
S_2=\sum_{j=1}^K\alpha_j^2,
\qquad
S_4=\sum_{j=1}^K\alpha_j^4.
\]
Then \(V_X=2qS_2+\sigma_U^2+\sigma_{\varepsilon_1}^2\), and the fourth cumulant of \(A\) is
\[
\kappa_A^{(4)}=2q(1-6q)S_4.
\]
The analogue of \eqref{eq:app-omega11-cumulant-decomp} is
\[
\omega_{11}
=
\omega_{11}^{\mathrm{cls}}
+
\frac{\kappa^2\kappa_A^{(4)}+(\lambda-\kappa)^2\kappa_U^{(4)}}{V_X^2}.
\]
Thus the robust--classical discrepancy is not confined to a single-SNP setting, although it is attenuated when the first-stage signal is spread evenly across many variants rather than concentrated in a few large effects.

The implication for the proposed statistic is direct. If the classical target is used in the numerator correction, the remaining first-order covariance with the screening statistic is
\[
\omega_{12}-\frac{\omega_{12}}{\omega_{11}^{\mathrm{cls}}}\omega_{11}
=
\omega_{12}\left(1-\frac{\omega_{11}}{\omega_{11}^{\mathrm{cls}}}\right),
\]
which is generally nonzero. The robust influence-function target in \eqref{eq:app-omega11-robust} is therefore required for the decorrelation coefficient.

\section{Supplement to Subsection~\ref{subsec:omega22_interpretation}: Full Delta-Method Variance of the Single-SNP Wald Ratio}
\label{app:wald-variance}

This appendix supplements Subsection~\ref{subsec:omega22_interpretation}. It derives the full first-order variance of the single-SNP Wald ratio and explains why a denominator-fixed standard error targets a different object away from the null.

Write the marginal SNP--exposure and SNP--outcome slopes as
\[
\hat\beta_{GX}=\frac{S_{GX}}{S_{GG}},
\qquad
\hat\beta_{GY}=\frac{S_{GY}}{S_{GG}},
\qquad
\hat\beta_{MR}=\frac{\hat\beta_{GY}}{\hat\beta_{GX}}=\frac{S_{GY}}{S_{GX}}.
\]
At the population level, \(\beta_{GX}=\alpha\), \(\beta_{GY}=\gamma\alpha\), and \(\beta_{GY}/\beta_{GX}=\gamma\). For the ratio map \(f(a,b)=a/b\), the delta method gives
\begin{align}
\Var(\hat\beta_{MR})
&\approx
\frac{\Var(\hat\beta_{GY})}{\beta_{GX}^2}
+
\frac{\beta_{GY}^2}{\beta_{GX}^4}\Var(\hat\beta_{GX})
-2\frac{\beta_{GY}}{\beta_{GX}^3}\Cov(\hat\beta_{GY},\hat\beta_{GX}).
\label{eq:app-wald-three-term-delta}
\end{align}
This is the first-order three-term variance expansion for a ratio estimator.

Equivalently, applying the ratio expansion directly to \(S_{GY}/S_{GX}\) yields
\[
\sqrt n(\hat\beta_{MR}-\gamma)
=
\frac1{\sqrt n}\sum_{i=1}^n
\frac{G_i(Y_i-\gamma X_i)}{\Cov(G,X)}+o_p(1).
\]
Because \(Y-\gamma X=\lambda U+\varepsilon_2\) is independent of \(G\),
\[
\omega_{22}^{\Delta}
=
\frac{\Var\{G(Y-\gamma X)\}}{\Cov(G,X)^2}
=
\frac{\lambda^2\sigma_U^2+\sigma_{\varepsilon_2}^2}{\alpha^2\sigma_G^2}.
\]
This is the full first-order variance factor used in the joint asymptotic theory.

A denominator-fixed approximation keeps only the first term in \eqref{eq:app-wald-three-term-delta}. Since the residual from the marginal regression of \(Y\) on \(G\) is
\[
Y-\beta_{GY}G
=Y-\gamma\alpha G
=(\gamma+\lambda)U+\gamma\varepsilon_1+\varepsilon_2,
\]
that approximation targets
\[
\omega_{22}^{\mathrm{pkg}}(\gamma)
=
\frac{(\gamma+\lambda)^2\sigma_U^2+\gamma^2\sigma_{\varepsilon_1}^2+\sigma_{\varepsilon_2}^2}{\alpha^2\sigma_G^2}.
\]
Therefore
\[
\omega_{22}^{\mathrm{pkg}}(\gamma)-\omega_{22}^{\Delta}
=
\frac{2\gamma\lambda\sigma_U^2+
\gamma^2(\sigma_U^2+\sigma_{\varepsilon_1}^2)}{\alpha^2\sigma_G^2}.
\]
At \(\gamma=0\), the two targets coincide. Away from the null, the denominator-fixed target can be either larger or smaller depending on the sign and magnitude of \(\gamma\lambda\). This affects standardization and power; it does not change the numerator decorrelation direction \(r=\omega_{12}/\omega_{11}\).

\section{Supplement to Subsection~\ref{subsec:omega22_interpretation}: IVW Variance and Summary-Data Implementations}
\label{app:ivw-package-variance}

This appendix supplements the \(\omega_{22}\) discussion in Subsection~\ref{subsec:omega22_interpretation}. It distinguishes the individual-level first-order variance of the IVW point estimator from the variance formulas used in summary data software implementations.

For uncorrelated variants, the IVW point estimator can be written as
\[
\hat\beta_{IVW}
=
\frac{\sum_{j=1}^K\pi_{j,n}\hat\beta_{Xj}\hat\beta_{Yj}}
{\sum_{j=1}^K\pi_{j,n}\hat\beta_{Xj}^2}.
\]
Appendix~\ref{app:joint-asymptotics} shows that, under the individual-level structural model,
\[
\sqrt n(\hat\beta_{IVW}-\gamma)
=
\frac1{\sqrt n}\sum_{i=1}^nT_{G,i}(Y_i-\gamma X_i)+o_p(1),
\]
where
\[
T_G=\sum_{j=1}^Kq_jG_j,
\qquad
q_j=\frac{\pi_j\alpha_j/\sigma_{G_j}^2}{\sum_{\ell=1}^K\pi_\ell\alpha_\ell^2}.
\]
Hence the individual-level first-order variance factor is
\[
\omega_{22}^{\mathrm{th}}
=(\lambda^2\sigma_U^2+\sigma_{\varepsilon_2}^2)
\frac{\sum_{j=1}^K\pi_j^2\alpha_j^2/\sigma_{G_j}^2}
{\left(\sum_{j=1}^K\pi_j\alpha_j^2\right)^2}.
\]
The ``squared weights over squared sum'' structure comes from the joint influence function of the collection of marginal SNP regressions when all are computed from the same individual-level sample.

A fixed-effect summary-data IVW standard error is often based on the weighted-regression approximation
\[
\Var_{\mathrm{FE}}(\hat\beta_{IVW})
\approx
\left(\sum_{j=1}^K \frac{\hat\beta_{Xj}^2}{\se(\hat\beta_{Yj})^2}\right)^{-1}.
\]
To compare this with \(\omega_{22}^{\mathrm{th}}\), let
\[
s_j^2=\Var(Y-\gamma\alpha_jG_j),
\qquad
\se(\hat\beta_{Yj})^2\approx \frac{s_j^2}{n\sigma_{G_j}^2}.
\]
Then inverse-variance weights have population limits proportional to \(\pi_j=\sigma_{G_j}^2/s_j^2\). With \(a_j=\alpha_j^2\sigma_{G_j}^2\), the individual-level first-order variance factor becomes
\[
\omega_{22}^{\mathrm{th}}
=
\tau^2
\frac{\sum_{j=1}^K a_j/s_j^4}
{\left(\sum_{j=1}^K a_j/s_j^2\right)^2},
\qquad
\tau^2=\lambda^2\sigma_U^2+\sigma_{\varepsilon_2}^2,
\]
whereas the fixed-effect summary-data variance factor is
\[
\omega_{22}^{\mathrm{FE}}
=
\left(\sum_{j=1}^K a_j/s_j^2\right)^{-1}.
\]
Their ratio is
\[
\frac{\omega_{22}^{\mathrm{th}}}{\omega_{22}^{\mathrm{FE}}}
=
\tau^2
\frac{\sum_{j=1}^K a_j/s_j^4}
{\sum_{j=1}^K a_j/s_j^2}.
\]
Under \(\gamma=0\), \(s_j^2=\tau^2\) for all \(j\), and the two variance factors coincide. Under alternatives, \(s_j^2\) can vary with the sign and magnitude of \(\gamma\), so the fixed-effect summary-data standard error is not uniformly conservative or uniformly anti-conservative relative to the individual-level first-order target. This explains the sign-dependent power differences reported in Section~\ref{subsec:omega22_interpretation}. As with the single-SNP case, these differences concern standardization and power rather than the decorrelation coefficient.

\bibliographystyle{unsrt}
\bibliography{reference}

@article{DaveySmithEbrahim2003,
  author  = {Davey Smith, George and Ebrahim, Shah},
  title   = {{Mendelian randomization}: can genetic epidemiology contribute to understanding environmental determinants of disease?},
  journal = {International Journal of Epidemiology},
  year    = {2003},
  volume  = {32},
  number  = {1},
  pages   = {1--22},
  doi     = {10.1093/ije/dyg070}
}

@article{DaveySmithHemani2014,
  author  = {Davey Smith, George and Hemani, Gibran},
  title   = {Mendelian randomization: genetic anchors for causal inference in epidemiological studies},
  journal = {Human Molecular Genetics},
  year    = {2014},
  volume  = {23},
  number  = {R1},
  pages   = {R89--R98},
  doi     = {10.1093/hmg/ddu328}
}

@article{Burgess2013SummarizedData,
  author  = {Burgess, Stephen and Butterworth, Adam and Thompson, Simon G.},
  title   = {Mendelian randomization analysis with multiple genetic variants using summarized data},
  journal = {Genetic Epidemiology},
  year    = {2013},
  volume  = {37},
  number  = {7},
  pages   = {658--665},
  doi     = {10.1002/gepi.21758}
}

@article{Hemani2018MRBase,
  author  = {Hemani, Gibran and Zheng, Jie and Elsworth, Benjamin and Wade, Kaitlin H. and Haberland, Valeriia and Baird, Denis and Laurin, Charles and Burgess, Stephen and Bowden, Jack and Langdon, Ryan and Tan, Vanessa Y. and Yarmolinsky, James and Shihab, Hashem A. and Timpson, Nicholas J. and Evans, David M. and Relton, Caroline and Martin, Richard M. and Davey Smith, George and Gaunt, Tom R. and Haycock, Philip C.},
  title   = {The {MR-Base} platform supports systematic causal inference across the human phenome},
  journal = {eLife},
  year    = {2018},
  volume  = {7},
  pages   = {e34408},
  doi     = {10.7554/eLife.34408}
}

@article{Wensley2011CRP,
  author  = {{C Reactive Protein Coronary Heart Disease Genetics Collaboration}},
  title   = {Association between {C} reactive protein and coronary heart disease: mendelian randomisation analysis based on individual participant data},
  journal = {BMJ},
  year    = {2011},
  volume  = {342},
  pages   = {d548},
  doi     = {10.1136/bmj.d548}
}

@article{Voight2012HDL,
  author  = {Voight, Benjamin F. and Peloso, Gina M. and Orho-Melander, Marju and Frikke-Schmidt, Ruth and Barbalic, Maja and Jensen, Majken K. and Hindy, George and Hólm, Hilma and Ding, Eric L. and Johnson, Toby and others},
  title   = {Plasma {HDL} cholesterol and risk of myocardial infarction: a mendelian randomisation study},
  journal = {The Lancet},
  year    = {2012},
  volume  = {380},
  number  = {9841},
  pages   = {572--580},
  doi     = {10.1016/S0140-6736(12)60312-2}
}

@article{Holmes2015Lipids,
  author  = {Holmes, Michael V. and Asselbergs, Folkert W. and Palmer, Tom M. and Drenos, Fotios and Lanktree, Matthew B. and Nelson, Christopher P. and Dale, Caroline E. and Padmanabhan, Sandosh and Finan, Chris and Swerdlow, Daniel I. and others},
  title   = {Mendelian randomization of blood lipids for coronary heart disease},
  journal = {European Heart Journal},
  year    = {2015},
  volume  = {36},
  number  = {9},
  pages   = {539--550},
  doi     = {10.1093/eurheartj/eht571}
}

@article{Meng2019VitaminD,
  author  = {Meng, Xiaomeng and Li, Xia and Timofeeva, Maria N. and He, Yuan and Spiliopoulou, Athina and Wei, Wen-Qing and Gifford, Angela and Wu, Hao and Varley, Tony and Joshi, Peter K. and others},
  title   = {Phenome-wide Mendelian-randomization study of genetically determined vitamin {D} on multiple health outcomes using the {UK Biobank} study},
  journal = {International Journal of Epidemiology},
  year    = {2019},
  volume  = {48},
  number  = {5},
  pages   = {1425--1434},
  doi     = {10.1093/ije/dyz182}
}

@article{Zhu2018CommonDiseases,
  author  = {Zhu, Zhihong and Zheng, Zhili and Zhang, Futao and Wu, Yang and Trzaskowski, Maciej and Maier, Robert and Robinson, Matthew R. and McGrath, John J. and Visscher, Peter M. and Wray, Naomi R. and Yang, Jian},
  title   = {Causal associations between risk factors and common diseases inferred from {GWAS} summary data},
  journal = {Nature Communications},
  year    = {2018},
  volume  = {9},
  pages   = {224},
  doi     = {10.1038/s41467-017-02317-2}
}

@article{BurgessThompson2011WeakInstruments,
  author  = {Burgess, Stephen and Thompson, Simon G.},
  title   = {Avoiding bias from weak instruments in Mendelian randomization studies},
  journal = {International Journal of Epidemiology},
  year    = {2011},
  volume  = {40},
  number  = {3},
  pages   = {755--764},
  doi     = {10.1093/ije/dyr036}
}

@article{Bowden2015MREgger,
  author  = {Bowden, Jack and Davey Smith, George and Burgess, Stephen},
  title   = {Mendelian randomization with invalid instruments: effect estimation and bias detection through {Egger} regression},
  journal = {International Journal of Epidemiology},
  year    = {2015},
  volume  = {44},
  number  = {2},
  pages   = {512--525},
  doi     = {10.1093/ije/dyv080}
}

@inproceedings{Huber1967,
  author    = {Huber, Peter J.},
  title     = {The Behavior of Maximum Likelihood Estimates Under Nonstandard Conditions},
  booktitle = {Proceedings of the Fifth Berkeley Symposium on Mathematical Statistics and Probability},
  volume    = {1},
  pages     = {221--233},
  publisher = {University of California Press},
  year      = {1967}
}

@article{White1980,
  author  = {White, Halbert},
  title   = {A Heteroskedasticity-Consistent Covariance Matrix Estimator and a Direct Test for Heteroskedasticity},
  journal = {Econometrica},
  volume  = {48},
  number  = {4},
  pages   = {817--838},
  year    = {1980}
}

@article{Fieller1954,
  author  = {Fieller, E. C.},
  title   = {Some Problems in Interval Estimation},
  journal = {Journal of the Royal Statistical Society. Series B (Methodological)},
  volume  = {16},
  number  = {2},
  pages   = {175--185},
  year    = {1954}
}

@article{Marees2018GWAS,
  author  = {Marees, Andries T. and de Kluiver, Hilde and Stringer, Sven and Vorspan, Florence and Curis, Emmanuel and Marie-Claire, Cynthia and Derks, Eske M.},
  title   = {A Tutorial on Conducting Genome-Wide Association Studies: Quality Control and Statistical Analysis},
  journal = {International Journal of Methods in Psychiatric Research},
  volume  = {27},
  number  = {2},
  pages   = {e1608},
  year    = {2018},
  doi     = {10.1002/mpr.1608}
}

@article{Uffelmann2021GWAS,
  author  = {Uffelmann, Emil and Huang, Qin Qin and Munung, Nchangwi Syntia and de Vries, Jantina and Okada, Yukinori and Martin, Alicia R. and Martin, Hilary C. and Lappalainen, Tuuli and Posthuma, Danielle},
  title   = {Genome-Wide Association Studies},
  journal = {Nature Reviews Methods Primers},
  volume  = {1},
  pages   = {59},
  year    = {2021},
  doi     = {10.1038/s43586-021-00056-9}
}

@misc{TwoSampleMRPackage,
  author       = {{MRC Integrative Epidemiology Unit}},
  title        = {{TwoSampleMR}: Two Sample Mendelian Randomization Functions and Interface to {MR-Base}},
  howpublished = {R package documentation and source code},
  year         = {2026},
  url          = {https://mrcieu.github.io/TwoSampleMR/},
  note         = {Accessed 2026-06-29}
}

@article{Freedman1983Screening,
  author  = {Freedman, David A.},
  title   = {A Note on Screening Regression Equations},
  journal = {The American Statistician},
  volume  = {37},
  number  = {2},
  pages   = {152--155},
  year    = {1983},
  doi     = {10.1080/00031305.1983.10482729}
}

@misc{FithianSunTaylor2014,
  author       = {Fithian, William and Sun, Dennis and Taylor, Jonathan},
  title        = {Optimal Inference after Model Selection},
  year         = {2014},
  eprint       = {1410.2597},
  archivePrefix = {arXiv},
  primaryClass = {math.ST},
  doi          = {10.48550/arXiv.1410.2597}
}

@article{TaylorTibshirani2015,
  author  = {Taylor, Jonathan and Tibshirani, Robert J.},
  title   = {Statistical Learning and Selective Inference},
  journal = {Proceedings of the National Academy of Sciences of the United States of America},
  volume  = {112},
  number  = {25},
  pages   = {7629--7634},
  year    = {2015},
  doi     = {10.1073/pnas.1507583112}
}

@article{LeeSunSunTaylor2016,
  author  = {Lee, Jason D. and Sun, Dennis L. and Sun, Yuekai and Taylor, Jonathan E.},
  title   = {Exact Post-Selection Inference, with Application to the Lasso},
  journal = {The Annals of Statistics},
  volume  = {44},
  number  = {3},
  pages   = {907--927},
  year    = {2016},
  doi     = {10.1214/15-AOS1371}
}

@book{Serfling1980Approximation,
  author    = {Serfling, Robert J.},
  title     = {Approximation Theorems of Mathematical Statistics},
  publisher = {John Wiley \& Sons},
  address   = {New York},
  year      = {1980},
  doi       = {10.1002/9780470316481}
}

@book{VanDerVaart1998,
  author    = {van der Vaart, Aad W.},
  title     = {Asymptotic Statistics},
  publisher = {Cambridge University Press},
  address   = {Cambridge},
  year      = {1998},
  doi       = {10.1017/CBO9780511802256}
}

@incollection{NeweyMcFadden1994,
  author    = {Newey, Whitney K. and McFadden, Daniel},
  title     = {Large Sample Estimation and Hypothesis Testing},
  booktitle = {Handbook of Econometrics},
  editor    = {Engle, Robert F. and McFadden, Daniel L.},
  volume    = {4},
  pages     = {2111--2245},
  publisher = {Elsevier},
  address   = {Amsterdam},
  year      = {1994},
  doi       = {10.1016/S1573-4412(05)80005-4}
}

@article{Hampel1974Influence,
  author  = {Hampel, Frank R.},
  title   = {The Influence Curve and Its Role in Robust Estimation},
  journal = {Journal of the American Statistical Association},
  volume  = {69},
  number  = {346},
  pages   = {383--393},
  year    = {1974},
  doi     = {10.1080/01621459.1974.10482962}
}

@book{Cochran1977,
  author    = {Cochran, William G.},
  title     = {Sampling Techniques},
  edition   = {3},
  publisher = {John Wiley \& Sons},
  address   = {New York},
  year      = {1977}
}

@article{Oehlert1992,
  author  = {Oehlert, Gary W.},
  title   = {A Note on the Delta Method},
  journal = {The American Statistician},
  volume  = {46},
  number  = {1},
  pages   = {27--29},
  year    = {1992},
  doi     = {10.1080/00031305.1992.10475842}
}

@manual{TwoSampleMRWaldRatio,
  author       = {{MRC Integrative Epidemiology Unit}},
  title        = {{mr\_wald\_ratio}: Perform 2 Sample IV Using Wald Ratio},
  organization = {{TwoSampleMR} package documentation},
  year         = {2026},
  url          = {https://mrcieu.github.io/TwoSampleMR/reference/mr_wald_ratio.html},
  note         = {Accessed 2026-07-02}
}

\end{document}